\documentclass[12pt]{article}

\usepackage{amscd,amsmath,amssymb,amsfonts,xspace,mathrsfs,amsthm,dsfont,bbm}
\usepackage{qtree}
\usepackage{color}

\usepackage{bbold}
\usepackage{latexsym}
\usepackage{graphicx}
\usepackage{dsfont}
\usepackage{color}
\usepackage{longtable}
\usepackage{hyperref}

\newlength{\abstractwidth}
  \newtheorem{lemma}{Lemma}
  \newtheorem{conj}{Conjecture}

\numberwithin{equation}{section}

\def\lfig#1#2#3#4#5{
\begin{figure}[t]
 \centerline{\includegraphics[width=#3]{#2}}
 \vspace{#5}
  \caption{#1 \label{#4}}
 \end{figure}
}

\def\bea{\begin{eqnarray}}
\def\eea{\end{eqnarray}}
\def\be{\begin{equation}}
\def\ee{\end{equation}}
\def\ba{\begin{align}}
\def\ea{\end{align}}
\def\bse{\begin{subequations}}
\def\ese{\end{subequations}}

\newcommand{\nn}{\nonumber}

\def\sign{{\rm sgn}}

\def\Im{\,{\rm Im}\,}
\def\Re{\,{\rm Re}\,}

\def\Ad{{\rm Ad}}
\def\ad{{\rm ad}}
\def\Aut{{\rm Aut}}

\def\({\left(}
\def\){\right)}
\def\[{\left[}
\def\]{\right]}
\def\<{\left\langle}
\def\>{\right\rangle}
\def\hf{{1\over 2}}

\newcommand{\p}{\partial}

\newcommand{\de}{\mathrm{d}}

\newcommand{\Yrm}{\mathrm{Y}}
\newcommand{\I}{\mathrm{i}}

\newcommand{\eps}{\epsilon}

\newcommand{\vth}{\vartheta}

\newcommand{\cB}{\mathcal{B}}

\newcommand{\cF}{\mathcal{F}}

\newcommand{\cC}{\mathcal{C}}

\newcommand{\cK}{\mathcal{K}}
\newcommand{\cM}{\mathcal{M}}

\newcommand{\cX}{\mathcal{X}}

\newcommand{\cP}{\mathcal{P}}

\newcommand{\cZ}{\mathcal{Z}}
\newcommand{\cI}{\mathcal{I}}

\newcommand{\cH}{\mathcal{H}}
\newcommand{\cU}{\mathcal{U}}

\def\scD{\mathscr{D}}

\def\scU{\mathscr{U}}
\def\scD{\mathscr{D}}
\def\Zv{\mathscr{Z}}

\def\Usf{\mathsf{U}}

\newcommand{\IT}{\mathds{T}}
\newcommand{\IR}{\mathds{R}}
\newcommand{\IC}{\mathds{C}}
\newcommand{\IZ}{\mathds{Z}}

\newcommand{\IN}{\mathds{N}}

\newcommand{\IP}{\mathds{P}}

\newcommand{\sgn}{\mbox{sgn}}

\newcommand{\tc}{\tilde c}

\def\tX{\tilde X}
\def\tT{\tilde T}

\def\cl0{\tilde c_0}

\def\bOm{\overline\Omega}

\def\ba{\bar a}

\def\by{\bar y}

\def\bZ{\bar Z}

\def\bt{\bar t}

\def\hx{\hat x}

\def\hK{\hat K}

\def\hPhi{\hat\Phi}
\def\hPsi{\hat\Psi}
\def\hpsi{\hat\psi}

\def\hcX{\hat\cX}
\def\hcU{\hat\cU}

\def\hUsf{\hat\Usf}
\def\hbcU{\hat{\overline{\cU}}}

\newcommand{\bfn}{{\boldsymbol n}}
\newcommand{\bfl}{{\boldsymbol l}}

\def\cXsf{\cX^{\rm sf}}

\def\hcXsf{\hat\cX^{\rm sf}}

\def\Phisf{\Phi^{\rm sf}}
\def\hPhisf{\hPhi^{\rm sf}}
\def\Phir{\Phi}

\def\cXr{\mathsf{X}^{\rm ref}}
\def\cXz{\mathsf{X}^{\rm sf}}
\def\tcXz{\mathsf{\tX}^{\rm sf}}
\def\cXb{\cX^{\rm b}}

\def\cXrp{\mathsf{X}^{\rm ref+}}
\def\cXrm{\mathsf{X}^{\rm ref-}}

\def\hcXr{\mathsf{X}}
\def\hcXrp{\mathsf{X}^{+}}
\def\hcXrm{\mathsf{X}^{-}}
\def\hcXrpm{\mathsf{X}^{\pm}}

\def\Kr{K^{\rm ref}}
\def\frKr{\mathfrak{K}}
\def\Wr{W^{\rm ref}}

\def\cHr{\cH^{\rm ref}}

\def\ker{\mathcal{K}}

\def\IS{\lefteqn{\textstyle\sum}\int}
\def\tIS#1{\resizebox{0.013\vsize}{!}{$\displaystyle{\IS_{#1}}$}}

\def\prodi#1#2{\lefteqn{\prod_{#1}^{#2}}\mathop{\phantom{\prod}}\nolimits}

\def\Gact{\Gamma_{\rm act}}

\begin{document}

\thispagestyle{empty}


\

\vskip 1cm

\begin{center}
{\Large \bf Quantum TBA for refined BPS indices}
\vskip 0.2in

{\large Sergei Alexandrov and Khalil Bendriss
}

\vskip 0.15in

$^\dagger$ {\it
Laboratoire Charles Coulomb (L2C), Universit\'e de Montpellier,
CNRS, F-34095, Montpellier, France}

%

\vskip 0.2in

\begin{abstract}
Refined BPS indices give rise to a quantum Riemann-Hilbert problem that is inherently related 
to a non-commutative deformation of moduli spaces arising in gauge and string theory compactifications.
We reformulate this problem in terms of a non-commutative deformation of a TBA-like equation 
and obtain its formal solution as an expansion in refined indices. As an application of this construction, 
we derive a generating function of solutions of the TBA equation in the unrefined case.
\end{abstract}

\vspace*{2mm} 
{\tt e-mail:
	{sergey.alexandrov@umontpellier.fr},
	{khalil.bendriss@umontpellier.fr}
}

\end{center}

\newpage
\tableofcontents

\baselineskip=13.5pt
\setcounter{page}{1}
\setcounter{equation}{0}
\setcounter{footnote}{0}
\setlength{\parskip}{0.15cm}

\section{Introduction}

The BPS spectrum provides invaluable information about any supersymmetric theory.
If the BPS states are counted with sign, one obtains BPS indices that are stable under deformations 
of parameters and thus provide a connection between weak and strong coupling regimes.
They encode the entropy of BPS black holes, determine instanton corrections to effective actions 
and coincide with various topological invariants of manifolds relevant to a given physical problem. 
Examples of such topological invariants captured by BPS indices include generalized Donaldson-Thomas (DT) invariants,
Vafa-Witten (VW) invariants, Seiberg-Witten (SW) invariants and many others.
Typically, they compute the Euler number of some moduli space such as the moduli space of some stable objects 
or instantons of some gauge group.

However, the Euler number is only one of the topological characteristics of a manifold.
A more refined information is given by its Betti numbers, which can be combined into Poincar\'e-Laurent polynomials 
to produce the {\it refined} version of topological invariants
\be
\label{defOmref}
\Omega(\gamma,y) = 
\sum_{p=0}^{2d} (-y)^{p-d}\, b_p(\cM_\gamma),
\ee
where $\cM_\gamma$ is the relevant moduli space characterized by charge $\gamma$, $d$ is its complex dimension,
and $y$ is the refinement parameter. In physics, this parameter arises as a fugacity conjugate to the angular momentum 
$J_3$ carried by BPS states in four dimensions and is switched on by
Nekrasov's $\Omega$-background \cite{Moore:1998et,Nekrasov:2002qd}, while $\Omega(\gamma,y)$ is called refined BPS index.

In contrast to the unrefined indices $\Omega(\gamma)$ obtained from \eqref{defOmref} at $y=1$, 
the refined indices are deformation invariant only in the presence of an R-symmetry.
It does exist in the cases of string compactifications on non-compact manifolds and supersymmetric gauge theories
\cite{Gaiotto:2010be}, 
but it is absent for string compactifications on compact Calabi-Yau (CY) threefolds. 
Nevertheless, even in such setups, the refined indices are expected to satisfy the standard wall-crossing relations
\cite{ks,Dimofte:2009bv,Dimofte:2009tm} and carry an important information about the theory.

A remarkable fact is that the refinement is inherently related to non-commutativity \cite{ks,Dimofte:2009bv}. 
In particular, the refined BPS indices define a quantum Riemann-Hilbert (RH) problem \cite{Barbieri:2019yya},
which is a natural non-commutative deformation of the RH problem set up by the unrefined DT invariants \cite{Bridgeland:2016nqw}.
The latter in turn is known to encode the geometry of some quaternionic spaces 
\cite{Bridgeland:2019fbi,Bridgeland:2020zjh,Alexandrov:2021wxu}, 
which in physics appear either as the moduli space of circle compactifications of 4d $N=2$ SYM \cite{Gaiotto:2008cd}
or as the hypermultiplet moduli space of CY string compactifications \cite{Alexandrov:2008gh,Alexandrov:2011va}.
Importantly, this classical RH problem can be reformulated as a system of integral TBA-like equations
\cite{Gaiotto:2008cd,Alexandrov:2009zh,Alexandrov:2010pp}, 
which have played an important role in recent advances in various fields 
\cite{Cecotti:2014zga,Ito:2017ypt,Kachru:2018van,Alexandrov:2018lgp,DelMonte:2022kxh}
and are similar to the TBA equations appearing in the context of AdS/CFT correspondence \cite{Alday:2009dv,Alday:2010ku}.
Therefore, it is natural to ask whether a similar reformulation exists for the quantum RH problem associated with 
the refined BPS indices.

In fact, at least two such reformulations have already appeared in the literature. 
First, a quantum version of the TBA introduced in \cite{Gaiotto:2008cd} was suggested in \cite{Cecotti:2014wea}.
However, as we show in appendix \ref{ap-qTBA}, though it seems very natural, 
it actually does {\it not} solve the quantum RH problem defined by the refined BPS indices.
Second, an integral equation involving the non-commutative Moyal star product
was put forward in \cite{Alexandrov:2019rth}. It was shown that a simple integral of its formal perturbative solution $\cXr_\gamma$
defines a potential that has two remarkable properties. On one hand, in the string theory context, 
it provides a refined version of the instanton corrected four-dimensional dilaton 
transforming as a Jacobi form under S-duality. On the other hand, in the pure mathematical framework of 
the so-called Joyce structures \cite{Bridgeland:2019fbi,Bridgeland:2020zjh}, it gives a solution to
a non-commutative deformation of the Plebanski heavenly equation \cite{Alexandrov:2021wxu}.
However, the proposed integral equation has a very unusual form, quite different from the standard TBA equations.
Furthermore, it does not have a well-defined unrefined limit, nor do its solutions!
Thus, it cannot be considered as a proper quantum deformation of the TBA solving the classical RH problem.

In this paper we revise this situation. We find that the solution $\cXr_\gamma$ of the integral equation 
introduced in \cite{Alexandrov:2019rth} can be used to define new functions $\hcXr_\gamma$ that 
do solve the quantum RH problem and reduce to solutions $\cX_\gamma$ of the classical TBA.
We obtain an explicit perturbative series for $\hcXr_\gamma$. 
However, we were not able to find any simple integral equation that it satisfies --- only equations 
involving infinitely many terms seem to be possible. 
Thus, the best way to define $\hcXr_\gamma$ is to pass through the integral equation defining $\cXr_\gamma$.

An interesting feature of $\hcXr_\gamma$ is that these functions for all charges $\gamma$ can be obtained 
by the adjoint action of a charge-independent potential, which can be considered as a kind of generating function for 
$\hcXr_\gamma$. Remarkably, its logarithm has a well-defined unrefined limit where it reduces to a generating function
of $\cX_\gamma$'s, for which we conjecture an explicit infinite series representation. 
This solves a long standing problem since this generating function was known only up to 
second order in the perturbative expansion \cite{Alexandrov:2017qhn}.
 
Finally, we show that, in the case of an uncoupled BPS structure, the functions $\hcXr_\gamma$
reproduce the solution to the quantum RH problem given in \cite{Barbieri:2019yya}.

The structure of the paper is the following. In the next section, we review 
the classical RH problem defined by BPS indices and its reformulation in terms of TBA-like equations.
In section \ref{sec-ref}, we introduce the quantum RH problem, 
non-commutative integral equation from \cite{Alexandrov:2019rth},
new functions $\hcXr_\gamma$ and describe their properties.
In section \ref{sec-genfun}, we derive generating functions for $\hcXr_\gamma$ and $\cX_\gamma$.
Then in section \ref{sec-case} we consider the case of an uncoupled BPS structure
and we conclude in section \ref{sec-conl}.
The appendices contain the analysis of the proposal from \cite{Cecotti:2014wea}
and proofs of various statements made in the main text.

\section{BPS indices and TBA}
\label{sec-TBA}

Let us assume that we have a set of BPS indices $\Omega(\gamma)$, characterized by charge $\gamma$, an element
of an even-dimensional lattice $\Gamma=\IZ^{2n}$, and satisfying the symmetry property $\Omega(-\gamma) =\Omega(\gamma)$.
In addition, the lattice $\Gamma$ is equipped with a skew-symmetric integer pairing $\langle\, \cdot\, ,\,\cdot\, \rangle$
and a linear map $Z\ :\ \Gamma\to \IC$, which is called central charge.
We also introduce BPS rays 
\be 
\ell_\gamma=\{t\in\IP^1 \ :\ Z_\gamma/t\in\I \IR^-\}.
\label{BPSray}
\ee 
Given a ray $\ell$, we denote
\be 
\Gamma_\ell=\{\gamma\in\Gamma_\star\ :\ \I Z_{\gamma}\in \ell, \ \Omega(\gamma)\ne 0\},
\label{setGl}
\ee  
where $\Gamma_{\!\star}=\Gamma\backslash\{0\}$, and call the ray active if $\Gamma_\ell$ is non-empty.
We assume that we are away from walls of marginal stability which means that, for any active BPS ray, 
if $\gamma,\gamma'\in \Gamma_\ell$ then they are mutually local, i.e. $\langle\gamma,\gamma'\rangle=0$.
This set of data then gives rise to the following Riemann-Hilbert problem, which is a slight generalization of the problem
introduced in \cite{Bridgeland:2016nqw}:

{\bf RH problem:} {\it  Find piece-wise holomorphic functions $\cX_\gamma(t)$ such that
\begin{enumerate}
	\item 
	$\cX_\gamma\cX_{\gamma'}=\cX_{\gamma+\gamma'}$;
	\item 
	their leading behavior at $t\to 0$ and $t\to \infty$, possibly up to a factor depending on the setup, 
	is captured by functions $\cXsf_\gamma(t)$;
	\item 
	they jump across active BPS rays in such a way that, if  
	$\cX_\gamma^\pm$ are values on the clockwise and anticlockwise sides of $\ell$, respectively,
	then they are related by the Kontsevich-Soibelman (KS) symplectomorphism 
	\be 
	\cX_\gamma^-=\cX_\gamma^+ \prod_{\gamma'\in\Gamma_\ell}
	\(1-\sigma_{\gamma'} \cX_{\gamma'}^+\)^{\Omega(\gamma')\langle\gamma',\gamma\rangle},
	\label{KSjump}
	\ee 
	where $\sigma_\gamma$ is a sign factor known as quadratic refinement and defined by the relation
	$\sigma_\gamma\, \sigma_{\gamma'} = (-1)^{\langle\gamma,\gamma'\rangle} \sigma_{\gamma+\gamma'}$.
\end{enumerate} 
}

The first condition means that there are only $2n$ independent functions that are associated with the generators of the lattice
and form a symplectic vector $\Xi(t)$ so that one can write $\cX_\gamma(t)=e^{2\pi\I \langle\gamma,\Xi(t)\rangle}$.
The second condition specifying the asymptotics at $t\to 0$ and $t\to \infty$ is presented intentionally 
in a vague form because there are actually various RH problems differing by the precise formulation of this condition.
To encompass them all, we had to be vague here. Nevertheless, in all setups described below
one finds that 
\be 
\cXsf_\gamma(t)=\cX_\gamma(t)|_{\Omega=0}.
\label{condOm0}
\ee 
We emphasize however that the condition \eqref{condOm0} cannot replace the second condition above since 
it is a weaker constraint on $\cX_\gamma$.
Finally, the third condition is the most important one since it introduces the dependence on the BPS indices
which control the jumps across BPS rays.

In fact, the functions $\cX_\gamma(t)$ have a geometric meaning playing the role of 
Darboux coordinates\footnote{More precisely, the Darboux coordinates are given by the symplectic vector $\Xi(t)$,
	but we will abuse the terminology and call $\cX_\gamma(t)$ in the same way.}
on the twistor space over a certain quaternionic moduli space $\cM$, where the variable $t$ parametrizes 
the $\IP^1$ fiber of the twistor fibration. The nature of $\cM$ depends on the concrete setup
and determines the exact asymptotic properties of  $\cX_\gamma(t)$. In particular, the parameters of $\cXsf_\gamma(t)$ 
can be seen as coordinates on $\cM$ and, as follows from \eqref{condOm0}, 
the functions $\cXsf_\gamma(t)$ themselves determine the geometry of $\cM$ 
in the case of vanishing BPS indices. Typically, they describe a torus fibration $T^{2n}\to \cM\to\cB$ with flat fibers, 
hence the index `sf' meaning `semi-flat'.

In all cases of interest, the above RH problem can be reformulated as a system of TBA-like equations.
These equations are simplified if one trades the integer BPS indices for their rational counterparts defined by
\be 
\bOm(\gamma) = \sum_{d|\gamma} \frac{1}{d^2}\, \Omega(\gamma/d).
\label{def-bOm}
\ee 
Then they can be written in the following concise form
\be
\cX_0=\cX_0^{\rm sf}
\exp\[\IS_1 K_{01} \cX_1 \].
\quad
\label{TBAeq-sh}
\ee
Here we used notations introduced in \cite{Alexandrov:2021wxu}:
$\cX_i=\cX_{\gamma_i}(t_i)$,
the sign \tIS{i} denotes a combination of the sum over charges $\gamma_i$ and the integral over 
the BPS ray $\ell_{\gamma_i}$ weighted by the rational indices $\bOm(\gamma)$, 
while $K_{ij}(t_i,t_j)$ is an integration kernel with a simple pole at $t_i=t_j$.
The idea is that crossing a BPS ray, one picks up the residue at the pole that produces the KS jump \eqref{KSjump}.
The replacement in the integrand of the usual factor $\log(1-\sigma_{\gamma_1}\cX_1)$ by $\cX_1$ is the effect of 
the use of the rational BPS indices.
The precise definition of \tIS{i} and $K_{ij}$ varies from case to case.
In this paper, we will be interested in the following three setups:

\paragraph{Setup 1.} 
This is the original RH problem introduced in \cite{Bridgeland:2016nqw}.
In this case $\cM$ is a complex hyperk\"ahler (HK) manifold with the triplet of symplectic two-forms $\omega_i$
determined by the Darboux coordinates $\Xi^a(t)$, components of $\Xi(t)$ in a basis $\gamma^a\in\Gamma$
($a=1,\dots, 2n$), through  
\be
\omega=\frac{t}{2}\sum_{a,b} \omega_{ab}\,  \de\Xi^a \,\de\Xi^b= t^{-1} \omega_+ -\I  \omega_3 + t \omega_-\, ,
\ee
where $\omega_{ab}$ is the inverse of the constant (integer-valued) skew-symmetric matrix 
$\omega^{ab}=\langle\gamma^a,\gamma^b\rangle$.
The asymptotic condition on $\cX_\gamma(t)$ (condition 2 in the above RH problem) requires that
at $t\to 0$ it reduces to 
\be
\cXsf_\gamma(t) =  e^{2\pi\I(\theta_{\gamma}-Z_\gamma/t)},
\qquad
\theta_\gamma=\langle \gamma,\theta\rangle,
\qquad
Z_\gamma=\langle \gamma,z\rangle,
\label{cXsf}
\ee 	
while at $t\to \infty$ it behaves polynomially in $t$. The $4n$ complex parameters of $\cXsf_\gamma$,
$z^a=Z_{\gamma^a}$ and $\theta^a=\theta_{\gamma^a}$, play the role of coordinates on $\cM$.
The asymptotic conditions lead to the following data for the TBA equation \eqref{TBAeq-sh} \cite{Alexandrov:2021wxu}
\be
\IS_i=\sum_{\gamma_i\in\Gamma_{\!\star} }
\frac{\sigma_{\gamma_i}\bOm(\gamma_i)}{2\pi\I}\int_{\ell_{\gamma_i}}\frac{\de t_i}{t_i^2}\, ,
\qquad
K_{ij}=\gamma_{ij}\,\frac{t_i t_j}{t_j-t_i}\, ,
\label{short3}
\ee
where $\gamma_{ij}=\langle\gamma_i,\gamma_j\rangle$.
	
\paragraph{Setup 2.} 
The second case describes either the HK moduli space of 4d $N=2$ gauge theory on $\IR^3\times S^1$ \cite{Gaiotto:2008cd}
or the D-instanton corrected quaternion-K\"ahler (QK) hypermultiplet moduli space of type II string theory on a CY threefold
\cite{Alexandrov:2008gh,Alexandrov:2011va}. Although the two manifolds are different, 
their twistorial descriptions are identical due to the so-called QK/HK correspondence 
\cite{Haydys,Alexandrov:2011ac,Hitchin:2012vvn}\footnote{Importantly, the BPS spectra are very different in the two cases.
In particular, in string theory BPS indices grow exponentially with charges and spoil a convergence property imposed on them 
in \cite{Bridgeland:2016nqw} to have a well-defined RH problem. There are indications \cite{Pioline:2009ia}
that the convergence issue might be resolved by inclusion of NS5-brane instanton corrections, 
which are still missing in this setup but do affect the hypermultiplet moduli space 
\cite{Becker:1995kb,Alexandrov:2010ca,Alexandrov:2014rca,Alexandrov:2023hiv}.}
and is determined by the Darboux coordinates $\cX_\gamma(t)$.
These coordinates are restricted to satisfy the following reality condition
\be 
\overline{\cX_\gamma(-1/\bt)}=\cX_{-\gamma}(t),
\label{real-cX}
\ee 
while the asymptotic condition in this case requires that both at $t\to 0$ and $t\to \infty$ they reduce to
\be
\cXsf_\gamma(t)= \exp\[2\pi\I\(\(\bZ_\gamma\,t-\frac{Z_\gamma}{t}\)
+\Theta_\gamma\)\],
\label{defXsf}
\ee
up to a real $t$-independent factor.
Here $\Theta_\gamma=\langle \gamma,\Theta\rangle$ is real and $Z_\gamma$ is parameterized by $n$ complex parameters
corresponding to a Lagrangian subspace, determined by a holomorphic prepotential,
in the space of Bridgeland stability conditions parametrized by all $z^a$.
Together they provide $4n$ real coordinates of $\cM$.\footnote{In the QK case, the counting is a bit different. 
$Z_\gamma$ is parametrized by $n-1$ complex and one real parameter and the missing real parameter 
is recovered from the asymptotic condition on an additional Darboux coordinate which, in the absence of NS5-brane instantons, 
is uniquely determined by $\cX_\gamma$.}
The data for the TBA equation \eqref{TBAeq-sh} are given by
\be
\IS_i=\sum_{\gamma_i\in\Gamma_{\!\star} }
\frac{\sigma_{\gamma_i}\bOm(\gamma_i)}{2\pi\I} \int_{\ell_{\gamma_i}}\frac{\de t_i}{t_i}\, ,
\qquad
K_{ij}=\frac{\gamma_{ij}}{2}\,\frac{t_j +t_i}{t_j-t_i}\, .
\label{short1}
\ee	
It is important to note that there is a relation between this setup and the previous one.
First, the kernels \eqref{short3} and \eqref{short1} differ by a constant term:
$\hf\,\frac{t'+t}{t'-t}=\frac{t}{t'-t}+\frac12$, which can be absorbed by  
a redefinition of the coordinates $\Theta\mapsto \theta$.
Note that, due to the reality condition \eqref{real-cX}, the absorbed term is pure imaginary 
so that the new coordinates $\theta$ are no longer real.
Then \eqref{cXsf} is obtained from \eqref{defXsf} in the so-called conformal limit \cite{Gaiotto:2014bza}:
$t\to 0$ and $Z_\gamma\to 0$ keeping the ratio $Z_\gamma/t$ fixed.
Simultaneously, one should relax the restriction on $Z_\gamma$ to belong to the Lagrangian subspace.

\paragraph{Setup 3.} 
The third case is a reduction of the previous setup. It describes 
the large volume limit of D3-instanton corrections to the hypermultiplet moduli space in type IIB string theory on a CY.
We refer to \cite{Alexandrov:2012au} for the details of the corresponding construction.
Here we restrict ourselves to presenting the data for the TBA equation \eqref{TBAeq-sh}
\be
\IS_i=\sum_{\gamma_i\in\Gamma_{\! +} }
\frac{\sigma_{\gamma_i}\bOm(\gamma_i)}{2\pi\I}\int_{\ell_{\gamma_i}}\de t_i,
\qquad
K_{ij}
=\frac{\gamma_{ij}}{t_j-t_i}+\I \, (vp_ip_j)\, , 
\label{short2}
\ee
where $v^\alpha$ are the K\"ahler moduli of the CY threefold, $p^\alpha$ are D3-brane charges
(the full charge vector is $\gamma=(p^0=0,p^\alpha,q_\alpha,q_0)$), 
$\Gamma_+$ is a positive cone inside $\Gamma_{\!\star}$ defined by the condition $(pv^2)>0$, and we use the notation 
$(xyz)=\kappa_{\alpha\beta\gamma}x^\alpha y^\beta z^\gamma$ where $\kappa_{\alpha\beta\gamma}$ are the 
triple intersection numbers of the CY.

\bigskip

The TBA equation \eqref{TBAeq-sh} can be solved by iterations. 
This results in a formal expansion\footnote{In this paper, we ignore the convergence issue of 
the expansion \eqref{Hexpand} and similar expansions appearing below. Anyway, it strongly depends on 
the Setup under consideration and its concrete realization. 
As the example of the resolved conifold demonstrates \cite{Bridgeland:2017vbr,Alexandrov:2021prq},
the issue arises even for uncoupled BPS structures with infinite BPS spectrum and 
requires a proper prescription to make sense of the perturbative expansions.
Similarly, it is hard to make any concrete statement about the uniqueness of solution 
of the classical RH problem and its quantum counterpart discussed in the next section
(see a related discussion in \cite{Bridgeland:2016nqw}). Typically, the uniqueness may require imposing 
additional analyticity properties and is expected to be closely related to the existence 
of a natural prescription defining the perturbative solution. 
}
in powers of BPS indices given by \cite{Filippini:2014sza}
\be
\label{Hexpand}
\cX_0 = \cXsf_0 \, \sum_{n=0}^{\infty} \left(
\prod_{i=1}^n  \IS_i\,  \cXsf_i \right) \sum_{T\in \IT_{n+1}^{\rm r}}\,
\frac{\prod_{e\in E_T}K_{s(e)t(e)}}{|\Aut(T)|}\, ,
\ee
where $\IT_n^{\rm r}$ is the set of rooted trees with $n$ vertices, each vertex
of the tree is decorated with a charge $\gamma_i\in\Gamma_{\!\star} $, with  $\gamma_0$ associated
to the root vertex, $E_T$ is the set of edges of $T$, 
and $s(e)$, $t(e)$ denote the source and target vertices of edge $e$, respectively. 
More explicitly, we get at the first few orders
\bea
\label{Xexpandshort}
\cX_0 &= &  \cXsf_{0} + \IS_1 K_{01} \cXsf_{0}\cXsf_{1}
+  \IS_1 \IS_2 \left( \tfrac12\, K_{01} K_{02} + K_{01} K_{12} \right)
\cXsf_{0}\cXsf_{1}\cXsf_{2}
\\
& + &  \IS_1\IS_2\IS_3\left( \tfrac16\, K_{01} K_{02} K_{03}
+ \tfrac12\, K_{01} K_{12} K_{13} + K_{01} K_{02} K_{13} + K_{01} K_{12} K_{23} \right)
\cXsf_{0}\cXsf_{1}\cXsf_{2}\cXsf_{3}+ \cdots .
\nn
\eea

Although this will not be needed in this work, it is worth mentioning that the function on $\cM$
defined by 
\be 
W= \IS_1 \cX_1\( 1-\hf\, \IS_2 K_{12} \cX_2\)
\label{Wfull-short}
\ee 
in Setup 1 was shown \cite{Alexandrov:2021wxu} to coincide with the Plebanski potential introduced in \cite{Bridgeland:2020zjh}
that satisfies the Pleba\'nski’s heavenly equation, while in Setup 3 it is identical to the instanton generating potential 
introduced in \cite{Alexandrov:2018lgp} that played a crucial role in establishing modular properties of the generating functions 
of D4-D2-D0 BPS indices in type IIA string theory on a CY threefold (see \cite{Alexandrov:2025sig} for a recent review).

\section{Refinement and non-commutativity}
\label{sec-ref}

\subsection{Refined BPS indices and quantum RH problem}
\label{subsec-refBPS}

As discussed in the Introduction, refined BPS indices $\Omega(\gamma,y)$ are Poincar\'e-Laurent polynomials 
in the refinement parameter $y$ reducing at $y=1$ to the ordinary BPS indices. 
Switching on $y$ induces a non-commutative structure \cite{ks,Gaiotto:2010be,Dimofte:2009bv,Cecotti:2014wea,Alexandrov:2019rth} 
which can be formalized as a quantum RH problem. 

To formulate it, let us split the functions $\cXsf_\gamma(t)$, capturing the geometry of $\cM$ at vanishing BPS indices 
when it is given by a semi-flat torus fibration, into two parts: $\cXsf_\gamma(t)=\cXb_\gamma(t) x_\gamma$ where
$x_\gamma$ is a $t$-independent part that contains all dependence on the angular variables parametrizing the torus, 
while the remainder $\cXb_\gamma(t)$ depends on $t$ and coordinates on the base of the fibration. 
For example, in Setup 1 of the previous section,
$x_\gamma=e^{2\pi\I\theta_{\gamma}}$ and $\cXb_\gamma(t)=e^{-2\pi \I Z_\gamma/t}$.
Then, following \cite{Barbieri:2019yya}, let $\hx_\gamma$ be generators 
of the quantum torus algebra $\IC_y[\IT]$
\be 
\hx_\gamma\hx_{\gamma'}=y^{\langle\gamma,\gamma'\rangle}\,\hx_{\gamma+\gamma'}
\label{qtorus-alg}
\ee 
and $\hPhisf(t)$ are automorphisms of this algebra defined by the functions $\cXb_\gamma(t)$ we have just introduced
\be 
\hPhisf(t)\ :\ \hx_\gamma\mapsto \cXb_\gamma(t) \hx_\gamma.
\label{aut-sf}
\ee 
Another set of automorphisms is given by the operators representing the quantum KS transformations
\be
\hat\scU_\gamma=\Ad_{\hcU_\gamma},
\qquad
\hcU_\gamma=\prod_{n\in\IZ} E_y\(y^n \sigma_\gamma \hx_{\gamma}\)^{\Omega_{n}(\gamma)},
\label{hviaU}
\ee
where $E_y(x)$ is the standard (sometimes called compact) quantum dilogarithm\footnote{There are different
	conventions in the literature,  e.g.  \cite{Barbieri:2019yya,Chuang:2022uey}
	define the quantum dilogarithm as $\mathbb{E}_q(x)=\prod_{k=0}^{\infty}(1-q^k x)$,
	which is related to our definition by $E_y(x)=[\mathbb{E}_{y^2}(xy)]^{-1}$.}
\be
E_y(x) := \exp\left[ \sum_{k=1}^{\infty} \frac{(x y)^k}{k(1-y^{2k})} \right] 
= \prod_{k=0}^{\infty}(1-x y^{2k+1})^{-1}, 
\qquad
x,y\in\IC,\ |y|<1,
\label{defEy}
\ee
and $\Omega_n(\gamma)$ are the Laurent coefficients of the refined BPS indices
\be
\Omega(\gamma,y)=\sum_{n\in \IZ}\Omega_n(\gamma)\,y^n.
\label{LaurentOm}
\ee
Then we ask:

{\bf qRH problem:} {\it  Find a piece-wise holomorphic function $\hPhi(t)$ valued in operators
	on the quantum torus algebra $\IC_y[\IT]$ such that
	\begin{enumerate}
		\item 
		for all $t\in\IP^1$, $\hPhi(t)\in \Aut \IC_y[\IT]$;
		\item 
		its leading behavior at $t\to 0$ and $t\to \infty$, as in the classical RH problem, 
		is captured by the automorphism $\hPhisf(t)$ \eqref{aut-sf};
		\item 
it jumps across active BPS rays in such a way that, if  
$\hPhi^\pm$ are values on the clockwise and anticlockwise sides of $\ell$, respectively,
then they are related by the quantum KS automorphism 
\be 
\hPhi^-=\hPhi^+ \circ\prod_{\gamma'\in\Gamma_\ell}\hat\scU_{\gamma'}.
\label{qKSjump}
\ee 
	\end{enumerate} 
}

To make this quantum RH problem closer to the formulation of its classical counterpart presented in the previous section, 
let us first simplify the product of 
the operators $\hcU_\gamma$ over colinear charges, which all correspond to the same BPS ray and hence all appear in the product
over charges defining the KS automorphism \eqref{qKSjump}. This can be done by using the rational refined BPS indices \cite{Manschot:2010qz}
\be
\label{defbOm}
\bOm(\gamma,y) = \sum_{d|\gamma} \frac{y-y^{-1}}{d(y^d-y^{-d})}\, \Omega(\gamma/d,y^d) ,
\ee
which reduce to \eqref{def-bOm} at $y=1$ and are known to simplify the refined wall-crossing relations 
\cite{ks,Manschot:2010qz,Alexandrov:2018iao}.
Substituting the first representation in \eqref{defEy} of the quantum dilogarithm into the definition \eqref{hviaU},
one finds 
\be 
\begin{split} 
\prod_{m=1}^\infty \hcU_{m\gamma}=&\, 
\exp\[-\sum_{m=1}^\infty \sum_{k=1}^\infty  \frac{\Omega(m\gamma,y^k)\, (\sigma_\gamma \hx_\gamma)^{mk}}{k(y^k-y^{-k})}\]
\\
=&\, \exp\[-\sum_{m=1}^\infty\sum_{d|\gamma}  \frac{\Omega(m\gamma/d,y^d)\, \sigma_{m\gamma} \hx_{m\gamma}}{d(y^d-y^{-d})}\]
=\prod_{m=1}^\infty \hbcU_{m\gamma}(\hx),
\end{split} 
\label{prod-hviaU}
\ee
where
\be 
\hbcU_\gamma(\hx)=\exp\(-\frac{\bOm(\gamma,y)\sigma_\gamma \hx_\gamma}{y-1/y}\).
\ee

Next, we define the quantum analogues of the Darboux coordinates $\cX_\gamma$,
\be 
\hcX_\gamma(t)=\hPhi(t)\[\hx_\gamma\],
\label{def-hcX}
\ee 
which satisfy the same commutation relation \eqref{qtorus-alg} as $\hx_\gamma$ since $\hPhi(t)$ is an automorphism.
In terms of these variables and taking into account \eqref{prod-hviaU}, 
the wall-crossing relation \eqref{qKSjump} takes the form
\be 
\hcX_\gamma^-=\(\prod_{\gamma'\in\Gamma_\ell}\Ad_{\hbcU_{\gamma'}(\hcX)} \)\hcX_\gamma^+\, .
\label{qKSjump1}
\ee 
Furthermore, using the Campbell identity $Ad_{e^X}=e^{\ad_X}$, it is easy to see that 
\be 
\Ad_{\hbcU_{\gamma'}(\hcX)} \hcX_\gamma=\sum_{k=0}^\infty 
\frac{1}{k!}\Bigl(\sigma_{\gamma'}\bOm(\gamma',y)\kappa(\langle\gamma,\gamma'\rangle )\Bigr)^k\hcX_{\gamma+k\gamma'},
\label{qKSjump2}
\ee 
where $\kappa(x)=\frac{y^{x}-y^{-x}}{y-y^{-1}}$.
This allows to rewrite the r.h.s. of \eqref{qKSjump1} as a simple multiplication from the left of $\hcX_\gamma^+$ by
\be 
\hUsf_{\ell,\gamma}=\prod_{\gamma'\in\Gamma_\ell}
\exp\Bigl(\sigma_{\gamma'}\bOm(\gamma',y)\,\kappa(\langle\gamma,\gamma'\rangle)\, 
y^{\langle\gamma,\gamma'\rangle}\hcX_{\gamma'} \Bigr).
\label{defUKS}
\ee 
As a result, the above quantum RH problem can be reformulated as follows: 

{\bf qRH problem:} {\it  Find functions $\hcX_\gamma(t)$ valued in $\IC_y[\IT]$ such that
	\begin{enumerate}
		\item 
		they satisfy
\be 
\hcX_\gamma\hcX_{\gamma'}=y^{\langle\gamma,\gamma'\rangle}\,\hcX_{\gamma+\gamma'};
\label{qtorus-algX}
\ee 
		\item 
		their leading behavior at $t\to 0$ and $t\to \infty$ is captured by $\cXb_\gamma(t) \hx_\gamma$;
		\item 
		they jump across active BPS rays in such a way that, if  
		$\hcX_\gamma^\pm$ are values on the clockwise and anticlockwise sides of $\ell$, respectively,
		then they are related by
		\be 
		\hcX_\gamma^-=\hUsf_{\ell,\gamma} \hcX_\gamma^+ .
		\label{qKSjump-X}
		\ee 
	\end{enumerate} 
}

\subsection{Star product}
\label{subsec-star}

Let us restrict for simplicity to Setup 1 of section \ref{sec-TBA}. 
We will return to other Setups in section \ref{subsec-other}.
A convenient way to realize the quantum torus algebra \eqref{qtorus-alg} is to use the non-commutative Moyal star product. 
Representing the refinement parameter as  $y=e^{2\pi\I\alpha}$,
for any two functions $f,g$ on $\cZ=\IP^1\times \cM$, we define
\be
f \star g =  f \exp\[ \frac{\alpha}{2\pi\I}\sum_{a,b}\omega_{ab}\,
\overleftarrow{\p}_{\!\theta^a}\overrightarrow{\p}_{\!\theta^b} \] g.
\label{starproduct}
\ee
It is easy to check that $x_\gamma=e^{2\pi\I\theta_{\gamma}}$ introduced in the previous subsection satisfy
the algebra \eqref{qtorus-alg} with respect to the star product
\be 
x_\gamma\star x_{\gamma'}=y^{\langle\gamma,\gamma'\rangle}\, x_{\gamma+\gamma'}
\label{qtorus-alg-x}
\ee
and thus can be identified with the generators $\hx_\gamma$. 
This identification 
\be 
\iota\ :\ \hx_\gamma\mapsto x_\gamma
\label{iota}
\ee 
allows us to represent all functions valued in $\IC_y[\IT]$ that appeared in the previous subsection
by usual complex functions provided all products are replaced by the star product.
In particular, we define 
\be 
\begin{split} 
\cXz_\gamma(t)=&\,\Phisf(t)[\hx_\gamma],
\qquad
\Phisf(t)= \iota\circ \hPhisf(t),
\\
\hcXr_\gamma(t)=&\,\Phir(t)[\hx_\gamma],
\qquad \qquad
\Phir(t)= \iota\circ \hPhi(t).
\end{split} 
\label{hom-sf}
\ee 
It is clear that in our case $\cXz_\gamma=\cXsf_\gamma$. 
As we will see in section \ref{subsec-other}, this relation will be slightly modified in Setup 3.
Note that the star product can be evaluated also for functions at different parameters $t$.
In particular, one finds
\be
\cXz_{\gamma}(t)\star\cXz_{\gamma'}(t') =
y^{\langle\gamma,\gamma'\rangle}\cXz_{\gamma}(t)\cXz_{\gamma'}(t').
\label{starXX}
\ee

Now we can give one more reformulation of the quantum RH problem in terms of the usual functions:

{\bf qRH problem:} {\it  Find functions $\hcXr_\gamma(t)$ such that
	\begin{enumerate}
		\item 
		they satisfy
		\be 
		\hcXr_\gamma(t)\star \hcXr_{\gamma'}(t)=y^{\langle\gamma,\gamma'\rangle}\,\hcXr_{\gamma+\gamma'}(t);
		\label{starXXh}
		\ee 
		\item 
		at small $t$ they are given by $\cXz_\gamma(t)$;
		\item 
		they jump across active BPS rays in such a way that, if  
		$\hcXrpm_\gamma$ are values on the clockwise and anticlockwise sides of $\ell$, respectively,
		then they are related by
		\be \,
		\hcXrm_\gamma=\biggl[\prodi{\gamma'\in\Gamma_\ell}{}_{\,\star}
		\exp_\star\Bigl(\sigma_{\gamma'}\bOm(\gamma',y)\,\kappa(\langle\gamma,\gamma'\rangle)\, 
		y^{\langle\gamma,\gamma'\rangle}\hcXr_{\gamma'} \Bigr)\,\biggr]\star  \hcXrp_\gamma,
		\label{qKSjump-Xstar}
		\ee 
		where $\prod_\star$ denotes the star product and $\exp_\star(x)=\sum_{n=0}^\infty x\star \cdots \star x/n! $.
	\end{enumerate} 
}
\noindent
Two comments are in order:
\begin{itemize}
	\item 
	Note that, whereas the commutation relation \eqref{starXX} holds for arbitrary $t$ and $t'$,
	the relation \eqref{starXXh} is imposed only at equal $t$'s.
	\item 
	Since the charges belonging to $\Gamma_\ell$ are mutually local, in \eqref{qKSjump-Xstar}
	all products inside the square brackets can be replaced by the ordinary ones 
	and $\exp_\star(x)$ reduces to the ordinary exponential.
	However, in Setup 3 considered in section \ref{subsec-other}, this replacement leads to additional
	$y$-dependent factors.
\end{itemize}

\subsection{Refined variables}
\label{subsec-refvar}

Having mapped the quantum RH problem to the space of usual functions, 
one can ask whether there is a quantum version of the TBA-like equations 
\eqref{TBAeq-sh} whose solutions coincide with the functions $\hcXr_\gamma(t)$.
Let us postpone the answer to this question to the next subsection and instead consider 
the following system of integral equations, first proposed in \cite{Alexandrov:2019rth}
and used in this context in \cite{Alexandrov:2021wxu}, 
\be
\cXr_0 =\cXz_0\star\(1+\IS_1 \Kr_{01}\cXr_1\),
\label{inteqH-star}
\ee
where the symbol \tIS{i} is the same as in section \ref{sec-TBA}, just with $\bOm(\gamma_i)$
replaced by $\bOm(\gamma_i,y)$, whereas the integration kernel is
taken to be  
\be
\Kr_{ij}=\frac{\ker_{ij}}{y-1/y}\, , 
\qquad
\ker_{ij}=\frac{t_i t_j}{t_j-t_i}\, .
\label{defkerref}
\ee
A nice feature of these equations is that they involve the star product.
On the other hand, they do not look like a deformation of the usual TBA.
Moreover, it is clear that the functions $\cXr_\gamma$ cannot solve the quantum RH problem simply because neither 
the equations \eqref{inteqH-star} nor the functions have a well-defined unrefined limit due to the pole
of the kernel $\Kr_{ij}$ at $y=1$.
Besides, instead of the standard commutation relation \eqref{starXXh}, they satisfy 
\be  
\cXr_{\gamma+\gamma'}(t)=y^{-\langle\gamma,\gamma'\rangle}\,\cXz_{\gamma}(t)\star\cXr_{\gamma'}(t),
\label{comm-Xr}
\ee
which trivially follows from the defining equation \eqref{inteqH-star}.

Nevertheless, the functions $\cXr_\gamma$ turn out to possess some magic: if one solves \eqref{inteqH-star} 
by the usual iterative procedure producing a formal power series in refined rational BPS indices, which takes the form
\be  
\cXr_0=\cXz_0\sum_{n=0}^\infty \(\prod_{k=1}^n\IS_k \Kr_{k-1,k}\,\cXz_k\) y^{\sum_{j>i=0}^n\gamma_{ij}},
\label{expXref}
\ee 
and substitutes this solution into the potential
\be 
\Wr = \IS_1\, \cXr_1,
\label{Wref}
\ee
then it can be shown that \cite{Alexandrov:2019rth,Alexandrov:2021wxu}
\begin{itemize}
	\item 
	$\Wr$ does have the unrefined limit where it reproduces the potential $W$ \eqref{Wfull-short};
	\item 
	it solves the $\star$-deformed version of the Pleba\'nski’s heavenly equation 
	introduced in \cite{Strachan:1992em,Takasaki:1992jf,Bridgeland:2020zjh};
	\item 
	an analogous construction in Setup 3 reproduces a potential transforming as a Jacobi form that captures
	the modular properties of the refined D4-D2-D0 BPS indices.
\end{itemize}
These observations suggest that the functions $\cXr_\gamma$ cannot be completely irrelevant.
Note also the simplicity of the above equations compared to the unrefined case:
i) in contrast to \eqref{Hexpand},  the iterative solution \eqref{expXref} does not involve any sum over trees 
and the effect of the Moyal product is completely captured by the last $y$-dependent factor;
ii) the potential $\Wr$ \eqref{Wref} is given by a single integral, whereas its unrefined version \eqref{Wfull-short}
involves a double integral term.

Before we proceed, we have to resolve an ambiguity appearing in the perturbative solution \eqref{expXref}.
It arises every time when $\arg (Z_{\gamma_{k-1}})=\arg (Z_{\gamma_k})$. In such case the BPS rays $\ell_{\gamma_{k-1}}$
and $\ell_{\gamma_k}$ overlap so that the integration contour for $t_k$ goes through the pole of $\Kr_{k-1,k}$.
Note that this problem does not arise in the unrefined case because, at generic point of the moduli space, 
the phases of the central charges align only for charges with 
vanishing pairing $\langle\gamma_{k-1},\gamma_k\rangle$ and all kernels $K_{ij}$ defined in section \ref{sec-TBA}
have residues at the poles proportional to $\gamma_{ij}$.
To resolve the ambiguity, we postulate\footnote{In fact, the prescription \eqref{sym-prescr}
	follows from the derivation of the potential $\Wr$, 
	playing the role of the refined version of the instanton generating potential in Setup 3,
	from a modular generalized theta series expansion \cite{Alexandrov:2019rth},
	and it is also implicitly used in the derivation of the properties of $\Wr$, mentioned below 
	\eqref{Wref}, in Setup 1 \cite{Alexandrov:2021wxu} where symmetrization plays a crucial role. 
	Any other prescription would spoil these properties 
	as well as the solution of the quantum RH problem presented below.}
that the r.h.s. of \eqref{inteqH-star} at $t_0\in\ell_{\gamma_0}$
is defined by means of a symmetric prescription, i.e. for $\gamma_1\in \Gamma_0 \equiv \Gamma_{\ell_{\gamma_0}}$
we take
\be  
\int_{\ell_{\gamma_1}}=\hf\,\lim_{\eps\to 0}\(\int_{e^{\I\eps}\ell_{\gamma_0}}+\int_{e^{-\I\eps}\ell_{\gamma_0}}\).
\label{sym-prescr}
\ee 
Then a similar symmetric prescription should be used in \eqref{expXref}.
The only difference is that now any number of consecutive integrals can be along the same contour.
In such case one should first regularize by moving all contours apart from each other and 
then symmetrize over their relative ordering in the angular direction. 

This symmetric prescription plays an important role in the derivation of the wall-crossing properties 
of  the functions $\cXr_\gamma$. It is done in appendix \ref{ap-wallcros}
and the result is given by the following formula
\be 
\cXrm_\gamma=\Usf_\ell\(y^{2\langle\gamma,\gamma'\rangle}\cXz_{\gamma'}\)\star \cXrp_\gamma,
\qquad
\Usf_\ell(x_{\gamma'})=\prodi{\gamma'\in\Gamma_\ell}{}_{\,\star}
\exp_\star\(\frac{\sigma_{\gamma'}\bOm(\gamma',y)}{y-1/y}\,x_{\gamma'} \) .
\label{qKSjump-Xr}
\ee 
Although it looks similar to \eqref{qKSjump-Xstar}, the crucial difference is that the exponential 
involves $\cXz_{\gamma'}$ and not the refined variable.
Thus, the commutation relation \eqref{comm-Xr} and the wall-crossing relation \eqref{qKSjump-Xr} 
both indicate that $\cXr_\gamma$ is in a sense `a half' of the solution of the quantum RH problem.
The full solution will be constructed in the next subsection.

\subsection{Solution}
\label{subsec-new}

We construct a solution of the quantum RH problem from the solution of the 
integral equation \eqref{inteqH-star} introduced above.
Let us define the function
\be 
\psi(t_0) =1+\IS_1\,\Kr_{01} \cXr_1
\label{Jref}
\ee 
that enters that integral equation. Note that it is independent of any charge.
Then we introduce a new set of functions:
\be 
\label{hcXranz}
\begin{split} 
	\hcXr_0=&\, \psi(t_0)_\star^{-1}\star \cXr_0
	\\
	=&\, \psi(t_0)_\star^{-1}\star \cXz_0 \star \psi(t_0),
\end{split}
\ee 
where the star index means that the inverse should be evaluated using the star product as 
$(1+x)^{-1}_\star=\sum_{n=0}^\infty(-1)^n x\star \cdots \star x $.
We claim that these functions solve the quantum RH problem defined above.

To show that this indeed the case, we have to check the three conditions imposed on $\hcX_\gamma$.
The first is the product relation \eqref{starXXh}.
It is trivial to check that it does hold:
\be  
\label{hcXrcomm}
\begin{split} 
	\hcXr_{\gamma}(t) \star \hcXr_{\gamma'}(t) =&\, 
	\psi_\star^{-1}\star \cXz_{\gamma}(t)\star \cXz_{\gamma'}(t) \star \psi
	\\
	=&\,
	y^{\langle\gamma,\gamma'\rangle}\psi_\star^{-1}\star \cXz_{\gamma+\gamma'}(t) \star \psi
	\\
	=&\, y^{\langle\gamma,\gamma'\rangle}\hcXr_{\gamma+\gamma'}(t).
\end{split} 
\ee 
The second condition is also obvious since at $t_0=0$ the kernel $\Kr_{01}$ \eqref{defkerref} vanishes 
and hence the function $\psi$ reduces to 1,
while they both become $t_0$-independent at large $t_0$.
Thus, it remains to prove only the wall-crossing relation \eqref{qKSjump-Xstar}.

To this end, note that due to \eqref{inteqH-star}, we have $\psi=\cXsf_{-\gamma}\star \cXr_\gamma$.
Therefore, the wall-crossing formula for $\cXr_\gamma$ \eqref{qKSjump-Xr} implies a similar relation for $\psi$:
\be 
\psi^-=\Usf_\ell(\cXz_{\gamma'})\star  \psi^+,
\label{qKSjump-psi}
\ee 
where the power of $y$ disappeared from the argument of $\Usf_\ell$ due to the commutation with $\cXsf_{-\gamma}$.
This result makes it straightforward to compute the effect of wall-crossing on $\hcXr_\gamma$.
From \eqref{hcXranz}, one finds
\be 
\begin{split}
\hcXrm_\gamma=&\,
(\psi^+)^{-1}\star \Usf_\ell^{-1}(\cXz_{\gamma'})\star \cXz_\gamma\star\Usf_\ell(\cXz_{\gamma'})\star \psi^+
\\
=&\, \Usf_\ell^{-1}(\hcXr_{\gamma'})\star \hcXrp_\gamma\star\Usf_\ell(\hcXr_{\gamma'}).
\end{split} 
\label{qKSjump-hXr}
\ee 
This result already reproduces the wall-crossing relation \eqref{qKSjump1} after applying the map $\iota$.
Performing the same manipulations as below that formula, one arrives to the wall-crossing relation in the form \eqref{qKSjump-Xstar}. 
This completes the proof that $\hcXr_\gamma(t)$ defined by \eqref{hcXranz} is a solution of the quantum RH problem.

\subsubsection{Perturbative expansion}
\label{subsec-qpert}

In principle, it is straightforward to generate a perturbative expansion for $\hcXr_\gamma$.
However, its form following directly from \eqref{hcXranz}
does not look particularly nice (see \eqref{expinsf}).
Nevertheless, in appendix \ref{ap-pert} we prove that the following key relation satisfied by the kernel \eqref{defkerref}
\be 
\label{cyc3}
\ker_{ij}\ker_{ik}=\ker_{ij}\ker_{jk}+\ker_{ik}\ker_{kj}
\ee 
allows us to rewrite it as
\be
\label{linhHr-comm}
\hcXr_0=\cXz_0+
\sum_{n=1}^\infty \[\prod_{k=1}^n \IS_k\ker_{k-1,k}\] 
\left\{\cdots\left\{\left\{\cXz_0,\cXz_1\right\}_\star,\cXz_2\right\}_\star\cdots ,\cXz_n\right\}_\star,
\ee
where 
\be
\{f,g\}_\star=\frac{f \star g-g\star f}{y-1/y}
\label{Mbracket}
\ee
and, as in \eqref{expXref},
the symmetric prescription explained in subsection \ref{subsec-refvar} is implied.
Since the star product of any functions $\cXz_\gamma$ at different arguments $t$ is known \eqref{starXX}, 
one can explicitly evaluate all commutators, which gives
\be  
\hcXr_0= \cXz_0\,\sum_{n=0}^\infty \prod_{j=1}^n \[\IS_j \kappa\(\sum_{i=0}^{j-1}\gamma_{ij}\) 
\ker_{j-1,j}\cXz_j\],
\label{compact-hcXr}
\ee 
where $\kappa(x)$ is defined below \eqref{qKSjump2}.

This representation immediately ensures that $\hcXr_\gamma$ has a well-defined unrefined limit.
Indeed, since $\lim_{y\to 1}\kappa(x)=x$, one gets
\be  
\mathop{\lim}\limits_{y\to 1} \hcXr_0= \cXz_0\,\sum_{n=0}^{\infty}\prod_{j=1}^n \[\IS_j \(\sum_{i=0}^{j-1}\gamma_{ij}\) 
\ker_{j-1,j}\cXz_j\] .
\label{compact-hcXr-lim}
\ee   
A remarkable fact, which we prove in appendix \ref{ap-unref}, is that the expansion 
\eqref{compact-hcXr-lim} is actually identical to the one in \eqref{Hexpand}.
On one hand, this proves that the unrefined limit of $\hcXr_\gamma$ coincides with the solution of 
the classical RH problem $\cX_\gamma$, as expected from a solution of the quantum RH problem.
On the other hand, this provides an alternative representation for $\cX_\gamma$ which, in contrast to \eqref{Hexpand},
involves only trees of the trivial topology --- linear trees for which every vertex has only one child.

\subsubsection{Reconstructing integral equations}
\label{subsec-newinteq}

The above results show that the functions $\hcXr_\gamma$ are the refined versions of the Darboux coordinates $\cX_\gamma$.
However, in contrast to the latter, they are defined in two steps: first one introduces $\cXr_\gamma$ through 
the somewhat unusual integral equation \eqref{inteqH-star} and only then one defines $\hcXr_\gamma$ by means of \eqref{hcXranz}.
It is natural to ask whether one can avoid the intermediate step and introduce $\hcXr_\gamma$ directly as 
solutions of some new integral equation.

A natural step to eliminate $\cXr_\gamma$ is to invert the relation \eqref{hcXranz}.
This can be done perturbatively in $\hcXr_\gamma$ and produces the following expansion
\be
\label{cXr-hcXr}
\cXr_0=\sum_{n=0}^\infty \[\prodi{k=1}{n}_\star \IS_k \Kr_{k-1,k}\hcXr_k  \] \star \hcXr_0,
\ee 
where we introduced the star products ordered in the descending and ascending order of the index:
\be
\prodi{k=1}{n}_\star x_k=x_n\star \cdots\star x_1,
\qquad
\prodi{k=1}{n}^\star x_k=x_1\star \cdots\star x_n.
\label{defprodstar}
\ee 
This implies
\be 
\psi(t_0) =\sum_{n=0}^\infty \[\prodi{k=1}{n}_\star \IS_k \Kr_{k-1,k}\hcXr_k\] .
\label{Jref-hX}
\ee 
Combining \eqref{Jref-hX} with the second line in \eqref{hcXranz}, one finds the following equation
\be 
\sum_{n=0}^\infty \[\prodi{k=1}{n}_\star\IS_k \Kr_{k-1,k}\hcXr_k\] \star\hcXr_0=
\cXz_0\star \sum_{n=0}^\infty \[\prodi{k=1}{n}_\star\IS_k \Kr_{k-1,k}\hcXr_k \].
\label{expeq}
\ee 
Unfortunately, this equation involves infinitely many terms of increasing order and, despite our efforts,
we have not been able to find a simple function that could produce it after a perturbative expansion.

It is possible to write down another integral equation, which also has the form of an expansion:
\be
\label{intgeqn-comm}
\hcXr_0 = \cXz_0 + \sum_{n=1}^{\infty}\[\prod_{k=1}^n\IS_k \ker_{k-1,k}\]
\left\{\cdots\left\{\left\{\cXz_0,\hcXr_n\right\}_\star,\hcXr_{n-1}\right\}_\star\cdots ,\hcXr_1\right\}_\star .
\ee
Its proof is more non-trivial and we relegate it to appendix \ref{ap-eq}.
Its advantage is that it makes manifest that $\hcXr_\gamma$ has a well defined unrefined limit
because the star bracket \eqref{Mbracket} then reduces to the ordinary Poisson bracket.
Note a difference with \eqref{linhHr-comm}: whereas there one starts by commuting $\cXz_0$ with $\cXz_1$
and continues up to $\cXz_n$, here one proceeds in the opposite order.

\subsection{Other Setups}
\label{subsec-other}

\paragraph{Setup 2.} 
In this case, the star product is defined as in \eqref{starproduct} with $\theta$ replaced by $\Theta$.
In particular, $\cXz_\gamma=\cXsf_\gamma$ satisfies the same 
relation \eqref{starXX}. The difference is hidden in the form of the kernel.
Its unrefined version \eqref{short1} suggests that, whereas $\Kr_{ij}$ can still be written as 
in \eqref{defkerref}, the most natural choice for the function $\ker_{ij}$ is given by $\ker_{ij}=\hf\,\frac{t_j+ t_i}{t_j-t_i}$.
Then the construction of the functions $\cXr_\gamma$ and $\hcXr_\gamma$ as well as the proof that 
the latter satisfy the first and third conditions of the quantum RH problem still go through without any changes.

The second condition is more involved because, to decide whether it is satisfied or not, we should first 
formulate it properly in the presence of the refinement parameter. Or we can simply find the conditions satisfied 
by the functions $\hcXr_\gamma$. It is easy to check that both $\cXr_\gamma$ and $\hcXr_\gamma$
are consistent with the reality condition similar to \eqref{real-cX} provided $y\in \IR$ or, more generally, 
\be 
\overline{\cXr_\gamma(-1/\bt;\by)}=\cXr_{-\gamma}(t;y),
\qquad
\overline{\hcXr_\gamma(-1/\bt;\by)}=\hcXr_{-\gamma}(t;y).
\label{real-cXr}
\ee 
On the other hand, the asymptotic conditions at $t\to 0$ and $t\to \infty$ 
satisfied by $\hcXr_\gamma$ appear to be significantly modified by the $y$-dependence.
In this limit, one finds 
\be 
\hcXr_\gamma(t)\approx \(1\pm \frac{\Wr}{2(y-1/y)}\)_\star^{-1}\star \cXz_\gamma(t)\star \(1\pm \frac{\Wr}{2(y-1/y)}\),
\label{asym-s2}
\ee 
where $\Wr$ is defined in \eqref{Wref} and the two signs correspond to $t\to 0$ and $\infty$, respectively.
For real $y$, $\Wr\in \I\IR$ so that the factors in the round brackets are not real, but mapped to each other
by the complex conjugation.
Thus, the asymptotics of $\hcXr_\gamma$ differs from $\cXz_\gamma$ not by a real factor but 
by the adjoint action of an operator with non-trivial reality properties.

A crucial difference compared to Setup 1 is that the functions $\hcXr_\gamma$
defined by \eqref{hcXranz} do {\it not} admit a smooth unrefined limit $y\to 1$!
The reason for this is that the relevant kernel $\ker_{ij}$
does {\it not} satisfy the identity \eqref{cyc3}. Instead, it fulfills
\be 
\ker_{ij}\ker_{ik}=\ker_{ij}\ker_{jk}+\ker_{ik}\ker_{kj}+\frac14\, .
\ee 
The additional constant term has drastic consequences because it spoils the derivation of the representation
\eqref{linhHr-comm} for the perturbative expansion and hence all subsequent results, including the existence 
of the unrefined limit.

Of course, this fact calls for looking for another solution that would have such a limit.
The key is the observation that, until one fixes the asymptotic conditions on $\hcXr_\gamma$, 
the above construction has a huge ambiguity obtained by the replacement of $\cXz_\gamma$ by\footnote{Note that this ambiguity 
	is absent in Setup 1 because it would change the asymptotic form of $\hcXr_\gamma(t)$ at small $t$
	which is fixed to be exactly $\cXz_\gamma(t)$.} 
\be 
\tcXz_\gamma(t) =\cF_\star^{-1}[\cXz_\gamma]\star\cXz_\gamma(t) \star \cF[\cXz_\gamma], 
\ee 
where $\cF[\cXz_\gamma]$ is a $t$-independent functional 
and the star index means that the inverse is obtained using the star product. Indeed, one can check that it 
does not affect the proof of the first and third conditions of the quantum RH problem
because $\tcXz_\gamma$ satisfies the same commutation relation \eqref{starXX} as $\cXz_\gamma$.
Furthermore, the reality property \eqref{real-cXr} also continues to hold provided $\cF[\cXz_\gamma]$ is real, 
while in the asymptotic expression \eqref{asym-s2} the same replacement $\cXz_\gamma\to\tcXz_\gamma$ should be made. 
It turns out that, imposing the existence of the unrefined limit for $\hcXr_\gamma$, one can fix 
the functional $\cF[\cXz_\gamma]$ order by order. We have done this calculation up to the sixth order  
and, provided the only allowed explicit $y$-dependence is through the factors $1/(y-y^{-1})$ associated to each integral, 
the result appears to be unique!
Furthermore, it miraculously recombines into a function of integrated $\cXr_\gamma$, namely
\be 
\cF[\cXz_\gamma]=\cC_\star\(\IS_1\,\frac{\cXr_1}{2(y-1/y)}\), 
\qquad
\cC(x)=1+\hf\,x^2+\frac38\, x^4+\frac{5}{16}\, x^6+\cdots,
\label{resFC}
\ee 
where the star index means that in the Fourier expansion of $\cC(x)$ all powers are evaluated using the star product. 
It is interesting that the first terms in the expansion of $\cC(x)$ coincide with the expansion of $(1-x^2)^{-1/2}$.
Therefore, it is tempting to suggest that these two functions are actually the same.
If this is true, one obtains an all order construction of a solution of the RH problem 
(with the asymptotic conditions modified by the refinement) that does have the unrefined limit.

However, it turns out that the unrefined limit of so defined $\hcXr_\gamma$ does {\it not} coincide with
the solution of the classical RH problem $\cX_\gamma$. The difference arises already at second order
and is actually expected due to the modified asymptotic conditions encoded by \eqref{asym-s2}.
Thus, we have not been able to produce a solution of the quantum RH problem that would reduce to $\cX_\gamma$ 
at $y\to 1$ in the second Setup.\footnote{In fact, one can argue that this is likely impossible. 
	The argument relies on the observation that, if one follows the same rules of the game as above, 
	the only way to cancel singularities in the unrefined limit is to obtain nested commutators as in \eqref{linhHr-comm}. 
	They necessarily lead to the nested structure of the skew-symmetric products $\gamma_{ij}$ as in \eqref{compact-hcXr-lim}.
	On the other hand, already the third term in the expansion \eqref{Xexpandshort} of $\cX_\gamma$ 
	with the kernel given in \eqref{short1} spoils this structure.}

\paragraph{Setup 3.} 
In this Setup the appropriate star product has been defined in \cite{Alexandrov:2019rth}.
It uses the representation of the refinement parameter  $y=e^{2\pi\I(\alpha-\tau\beta)}$ with $\alpha,\beta\in\IR$ and 
$\tau$ the axio-dilaton field, and involves three additional vector valued fields: 
RR-axions $\tc_\alpha$ and $c^\alpha$ coupled to D3 and D1-branes, respectively, and the periods of the B-field $b^\alpha$.
Then instead of \eqref{starproduct}, one takes
\be
f \star g =  f \exp\[ \frac{1}{2\pi\I}\( \overleftarrow{\scD}_{\!\!\alpha}\overrightarrow{\p}_{\!\tc_\alpha}-
\overleftarrow{\p}_{\!\tc_\alpha} \overrightarrow{\scD}_{\!\! \alpha}\) \] g,
\label{starproduct-alt}
\ee
where $\scD_\alpha=\alpha\p_{c^\alpha}+\beta\p_{b^\alpha}$.
The advantage of this definition is that it preserves modular invariance which is an isometry
of the hypermultiplet moduli space originating from S-duality of type IIB string theory.\footnote{Modular transformations 
	of physical fields should be supplemented by a transformation of the refinement parameter. Namely $(\alpha,\beta)$ 
	should transform as a doublet under $SL(2,\IZ)$ which implies that $\log y$ transforms as an elliptic parameter of Jacobi forms.}
	
With respect to this star product, the semi-flat Darboux coordinates appearing in this setup 
satisfy a more involved commutation relation than \eqref{starXX}:
\be
\cXsf_{\gamma}(t)\star\cXsf_{\gamma'}(t')=y^{\frac{\beta}{2}\, (pp' (p+p'))} \,
\Yrm(\gamma,\gamma',t-t')
\,\cXsf_{\gamma}(t)\cXsf_{\gamma'}(t').
\label{starXX3}
\ee
where
\be 
\Yrm(\gamma,\gamma',\Delta t)=
y^{\langle\gamma,\gamma'\rangle}(y\by)^{-\I(vpp')\Delta t}.
\label{def-comY}
\ee
One could think of the additional factors as arising due to the non-vanishing real part of $\log y$.
If it is set to zero, i.e. $y$ is a pure phase, then $\beta=0$ and the relation \eqref{starXX3} reduces to \eqref{starXX}.
If one splits $\cXsf_\gamma(t)$ into a $t$-dependent part $\cXb_\gamma(t)$ and the angular part $x_\gamma$,
the star product of the latter still contains an additional factor,
\be 
x_\gamma\star x_{\gamma'}=y^{\langle\gamma,\gamma'\rangle+\frac{\beta}{2}\, (pp' (p+p'))}\, x_{\gamma+\gamma'},
\label{prod-xg3}
\ee 
which does not allow to make the identification \eqref{iota}.
However, the mismatch between the algebras \eqref{qtorus-alg} and \eqref{prod-xg3} can easily be rectified 
by multiplying $x_\gamma$ by an additional factor. Indeed, let $m(p)=\frac{1}{6}\,(p^3)+\rho_\alpha p^\alpha$ 
where $\rho_\alpha$ are arbitrary constant coefficients.
Then the following map 
\be 
\iota\ :\ \hx_\gamma\mapsto y^{\beta m(p)} x_\gamma
\label{iota3}
\ee 
is a homomorphism of the algebras.
Using this map in \eqref{hom-sf}, one obtains that, in contrast to Setup 1, now one has 
$\cXz_\gamma=y^{\beta m(p)}\cXsf_\gamma$. At equal parameters $t$, these variables satisfy the standard quantum torus algebra
\be  
\cXz_{\gamma}(t)\star\cXz_{\gamma'}(t)=y^{\langle\gamma,\gamma'\rangle}\,\cXz_{\gamma+\gamma'}(t),
\label{starXX3z}
\ee 
while at non-equal $t$'s their star product gives rise to the same factor $\Yrm(\gamma,\gamma',t-t')$ as in \eqref{starXX3}. 
Note that with respect to the usual product their algebra is also non-trivial: 
$\cXz_{\gamma+\gamma'}=y^{\frac{\beta}{2}\, (pp' (p+p'))}\cXz_{\gamma}\cXz_{\gamma'}$.

The kernel in this case is given as in \eqref{defkerref} with
\be 
\ker_{ij}=\frac{1}{t_j-t_i} 
\ee 
and can be checked to satisfy the identity \eqref{cyc3}.
As a result, all the equations in sections \ref{subsec-star}-\ref{subsec-new} that
are written in terms the star product should still hold.
The only equations that change are those where this product has been evaluated explicitly, 
\eqref{expXref} and \eqref{compact-hcXr}.
In the former the factor $y^{\sum_{j>i=0}^n\gamma_{ij}}$ is replaced by 
\be 
y^{\beta\(m\(\sum_{i=0}^n p_i\)-\sum_{i=0}^n m(p_i)\)}\prod_{j>i=0}^n\Yrm_{ij},
\ee 
where we introduced $\Yrm_{ij}=\Yrm(\gamma_i,\gamma_j,t_i-t_j)$,
while in the latter the factor $\kappa\(\sum_{i=0}^{j-1}\gamma_{ij}\)$ is replaced by
\be 
\frac{y^{\beta\(m\(\sum_{i=0}^j p_i\)-m\(\sum_{i=0}^{j-1} p_i\)-m(p_j)\)}}{y-1/y}
\(\prod_{i=0}^{j-1}\Yrm_{ij}-\prod_{i=0}^{j-1}\Yrm_{ji}\).
\ee 
Besides, the unrefined limit \eqref{compact-hcXr-lim}, which provides an alternative representation of 
$\cX_\gamma$, is also affected. It should be evaluated by first setting $\by=1$ and then taking the limit $y\to 1$
\cite{Alexandrov:2019rth}. As a result, the factor $\sum_{i=0}^{j-1}\gamma_{ij}$ is replaced by
\be
\sum_{i=0}^{j-1}\Bigl(\gamma_{ij}+\I(vp_ip_j)(t_j-t_i)\Bigr)=\sum_{i=0}^{j-1}(t_j-t_i) K_{ij},
\ee
where we expressed the result in terms of the kernel \eqref{short2}.
Despite this change, the proof presented in appendix \ref{ap-unref} that the resulting representation of $\cX_\gamma$ coincides 
with \eqref{Hexpand} still works with minor changes.
Most importantly, the function $\hcXr_\gamma$ defined by \eqref{hcXranz} and \eqref{inteqH-star} 
(or by one of the equations \eqref{expeq}, \eqref{intgeqn-comm}) still provides a solution of the quantum RH problem 
specified to this Setup.

\section{Adjoint representation and generating functions}
\label{sec-genfun}

As was noticed in \cite{Barbieri:2019yya}, the solution of the quantum RH problem can be represented in the adjoint form.
Namely, let us consider the map $\hPsi=\hPhi\circ (\hPhisf)^{-1}$. The claim is that there is a function $\hpsi(t)$
valued in $\IC_y[\IT]$ such that $\hPsi(t)=\Ad_{\hpsi(t)^{-1}}$.
Applying the identification \eqref{iota}, we get the function $\psi=\iota[\hpsi]$ and
the map $\Psi=\iota\circ \hPsi\circ\iota^{-1}$ such that
\be
\Psi=\Ad_{\psi^{-1}}\ :\ \cXz_\gamma  \mapsto \hcXr_\gamma,
\label{mapPsi}
\ee
where multiplication and inversion are evaluated using the star product.
Remarkably, the solution $\hcXr_\gamma$ \eqref{hcXranz} found in the previous section 
by construction has this adjoint form with $\psi$ defined in \eqref{Jref}.
Thus, one obtains the adjoint form of the solution for free.
Furthermore, as was emphasized above, $\psi(t)$ is independent of charges.
Thus, it can be considered as a {\it generating function} of the set of functions $\hcXr_\gamma(t)$ for different charges.

This result has a nice application: by taking the unrefined limit, 
it can be used to derive a generating function of Darboux coordinates $\cX_\gamma$.
Before, such a generating function was known only up second order in the perturbative expansion \cite{Alexandrov:2017qhn}
and a generalization to higher orders did not seem to be obvious.
Now we will find it to all orders.

To this end, let us apply the Campbell identity to the map \eqref{mapPsi}.
It allows us to rewrite it as 
\be 
\hcXr_\gamma=e^{-\ad_{\log_\star\!\psi}} \cXz_\gamma=e^{-\{ \cHr,\, \cdot\, \}_\star}\cXz_\gamma,
\label{genHref}
\ee 
where $\log_*(1+x)=\sum_{n=1}^\infty \frac{(-1)^{n-1}}{n} x\star \cdots \star x$, 
the star bracket was defined in \eqref{Mbracket} and we introduced 
\be 
\cHr=(y-y^{-1})\log_\star \psi.
\label{def-genref}
\ee 
In the unrefined limit $y\to 1$, $\hcXr_\gamma$ and $\cXz_\gamma$ reduce to $\cX_\gamma$ and $\cXsf_\gamma$, respectively,
while the star bracket becomes the ordinary Poisson bracket. As a result, the relation
\eqref{genHref} takes the form
\be 
\cX_\gamma=e^{-\{\cH,\, \cdot\, \}} \cXsf_\gamma,
\label{cXPb}
\ee 
where 
\be 
\cH=\lim_{y\to 1} \cHr.
\ee

Since $\psi$ does {\it not} have a well-defined unrefined limit, one could worry that $\cHr$ 
does not have one either. However, miraculously, it does! In appendix \ref{ap-gen} we prove that
\be
\label{logcJr2}
\cHr(t_0) = \sum_{n=1}^{\infty} \frac{1}{n}\[\prod_{k=1}^{n}\IS_{k} \ker_{k-1,k}\] \sum_{\sigma \in S_n}  \frac{(-1)^{s_\sigma+1}}{\binom{n}{s_\sigma}s_\sigma } 
\left\{\cdots\left\{\cXz_{\sigma(1)},\cXz_{\sigma(2)}\right\}_\star\cdots ,\cXz_{\sigma(n)}\right\}_\star,
\ee
where $S_n$ is the group of permutations of $n$ elements 
and $s_\sigma$ is equal to the number of ascending runs\footnote{An ascending run of a permutation is a nonempty increasing contiguous subsequence that cannot be extended at either end. Their number corresponds to $1+$ the number of elements $1\leq i \leq n-1$ such that $\sigma(i)>\sigma(i+1)$, known as the number of descents of $\sigma$.} of the permutation $\sigma$.
This representation makes the existence of the unrefined limit manifest
and it is straightforward to write the resulting expression for the generating function $\cH$.
Furthermore, it turns out that this expression can also be rewritten in terms of 
the functions $\cX_\gamma$. Namely, we propose the following
\begin{conj} 
\be
	\cH(t_0) = \IS_1 \ker_{01}\cX_1
	-\sum_{n=2}^{\infty}\frac{1}{n(n-1)} \[\prod_{k=1}^{n}\IS_{k}\,\ker_{k-1,k}\cX_{k}\]\(\sum_{\sigma \in S_n^{(+)}}\prod_{k=2}^{n}
	\gamma_{\sigma(k-1)\sigma(k)}\, \),
\ee
where $S_n^{(+)}$ is the set of permutations of $n$ elements such that 
$\sigma(\lfloor\frac{n}{2}\rfloor)<\sigma(\lceil\frac{n}{2}+1\rceil)$.
\end{conj}
\noindent
Although we do not have a proof of this conjecture, we performed its extensive numerical checks
which makes us confident that it should be true. 
It provides a generalization of the second order result from \cite{Alexandrov:2017qhn}
to all orders.

\section{Uncoupled case}
\label{sec-case}

Let us consider the simplest case of a BPS structure which is
\begin{itemize}
	\item  {\it finite}, i.e. $\Omega(\gamma,y)=0$ for all but finitely many $\gamma\in\Gamma$ 
	(we will call the set of such active charges $\Gact$);
	\item {\it uncoupled}, i.e. $\langle \gamma,\gamma'\rangle=0$ for any $\gamma,\gamma'\in\Gact$;
	\item {\it palindromic}, i.e. $\Omega_n(\gamma) = \Omega_{-n}(\gamma)$ for all $n \in \IZ$ and $\gamma \in \Gamma$;
	\item {\it integral}, i.e. $\Omega_n(\gamma) \in \IZ$ for all $n \in \IZ$ and $\gamma \in \Gact$.
\end{itemize}
In Setup 1, it has already been analyzed in \cite{Barbieri:2019yya} and here we show that the solution \eqref{compact-hcXr}
reproduces the one found in that work.

By a symplectic rotation of the basis of charges, for any uncoupled BPS structure, one can always choose it in such a way
that $\Gact$ is a sublattice of the lattice of electric charges $(0,q_\Lambda)\in \Gamma_e\subset \Gamma$. 
Hence, only such charges appear in the sums encoded in the symbols \tIS{i}.
Since the mutual Dirac products of electric charges all vanish, the result \eqref{compact-hcXr} implies that 
$\hcXr_{(0,q_\Lambda)}=\cXsf_{(0,q_\Lambda)}$, where we used that $\cXz_\gamma=\cXsf_\gamma$ in Setup 1.
For a generic charge, one finds instead
\be  
\begin{split} 
	\hcXr_0=&\, \sum_{n=0}^\infty \prod_{k=1}^n \[\IS_k \kappa(\gamma_{0k}) 
	\ker_{k-1,k}\cXsf_k\] \cXsf_0
	\\
	=&\,  \sum_{n=0}^\infty \frac{1}{n!}\prod_{k=1}^n \[\IS_k \kappa(\gamma_{0k}) 
	\ker_{0k}\cXsf_k\] \cXsf_0
	\\
	= &\,\cXsf_0\exp\[\IS_1 \kappa(\gamma_{01}) \ker_{01}\cXsf_1\],
\end{split}
\label{hcXr-mlocal}
\ee
where $\gamma_{0k}=-p_0^\Lambda q_{k,\Lambda}$ and the second equality is a consequence of Lemma \ref{lemma-tree}
applied to $T$ given by a tree shown in Fig. \ref{fig-trees-pr1}(a).
The expression in the exponential in \eqref{hcXr-mlocal} can be written explicitly as
\be
	S_0= \frac{1}{2\pi \I}
	\sum_{\gamma_1\in\Gamma_e}\sigma_{\gamma_1}\bOm(\gamma_1,y)\kappa(\gamma_{01})
	\int_{\ell_1} \frac{\de t_1}{t_1}\,\frac{t_0}{t_1-t_0}\,\cXsf_1.
\ee
Substituting the definition of rational BPS indices \eqref{defbOm} and the Laurent expansion of the integer valued ones 
\eqref{LaurentOm}, one obtains 
\be
\begin{split}
	S_0
	=& \,\frac{1}{2\pi \I }
	\sum_{\gamma_1\in\Gamma_e}\sum_{d|\gamma_1} \sum_{n\in\IZ}
	\sigma_{\gamma_1}\,\Omega_n(\gamma_1/d)\,
	y^{dn}\,\frac{\kappa(\gamma_{01})}{d \kappa(d)}
	\int_{\ell_1} \frac{\de t_1}{t_1} \,\frac{t_0}{t_1-t_0}\,\cXsf_1
	\\
	=& \,  \frac{1}{2\pi \I }
	\sum_{\gamma_1\in\Gact}\sum_{n\in \IZ}
	 \Omega_{n}(\gamma_1)
	\int_{\ell_1} \frac{\de t_1}{t_1}\,\frac{t_0}{t_1-t_0}
	\sum_{d=1}^{\infty}  y^{dn}\, \frac{\kappa(d\gamma_{01})}{d \kappa(d)}
	\,(\sigma_{\gamma_1}\cXsf_1)^{d}.
\end{split}
\ee
Since $\frac{\kappa(d\ell)}{\kappa(d)}=\sgn(\ell)\sum_{k=1-|\ell|}^{|\ell|-1} y^{dk}$, 
the sum over $d$ produces the logarithm. Then after substituting the result into \eqref{hcXr-mlocal} 
and taking into account that $\Omega_n(\gamma)=\Omega_{-n}(-\gamma)$, one remains with 
\be   
\hcXr_0=\cXsf_0 \prod_{\substack{\gamma_1\in\Gact	\\ \Re(Z_{\gamma_1})>0}}
\prod_{k=1-|\gamma_{01}|}^{|\gamma_{01}|-1}\prod_{n\in\IZ}
e^{s_{01}\Omega_n (\gamma_1)\cI_{\gamma_1}(t_0;y^{k+n}) },
\label{hcXr-mlocal-res}
\ee
where $s_{01}=\sign(\gamma_{01})$ and
\be 
\cI_{\gamma}(t;y)=\frac{1}{2\pi\I}\sum_{\eps=\pm}\eps
\int_{\ell_{\eps\gamma}}\frac{\de t'}{t'}\, \frac{t}{t-t'}\,\log\(1-y^\eps\sigma_{\eps\gamma}\cXsf_{\eps\gamma}(t')\).
\ee 
Substituting $y=e^{2\pi\I\alpha}$, $\cXsf_\gamma$ from \eqref{cXsf} and changing the integration variable 
$t'=\I Z_{\eps\gamma}/(s+\I(\vth_{\eps\gamma}+\alpha))$
where $e^{2\pi\I \vth_\gamma}=\sigma_\gamma e^{2\pi\I\theta_\gamma}$, one obtains
\be 
\cI_{\gamma}(t;y)=\frac{1}{2\pi\I}\sum_{\eps=\pm}\eps
\int_{-\I\eps\vth_\gamma}^\infty\de s\, \frac{\log(1-e^{-2\pi s})}{s-\I\eps(Z_\gamma/t-\vth_\gamma-\alpha)}\,.
\ee 
Exactly this integral, but with $\alpha$ set to zero, has been computed in \cite[\S C]{Alexandrov:2021wxu}.
The result is given by
\be 
\cI_{\gamma}(t;y)=-\log\Lambda\(\frac{Z_{\gamma}}{t},1-\[\vth_{\gamma}+\alpha \]\),
\ee 
where 
\be
\label{defLambda}
\Lambda(z,\eta) =\frac{e^z\,\Gamma(z+\eta)}{\sqrt{2\pi} z^{z+\eta-1/2}}
\ee
and
\be
[x] = 
\begin{cases}
	x - \lfloor \Re(x)\rfloor, &\text{if } \Im x \geq 0,
	\\
	x + \lfloor \Re(-x)\rfloor+1,\quad &\text{if } \Im x <0.
\end{cases}
\ee
Substituting this result into \eqref{hcXr-mlocal-res}, one finally arrives at 
the following solution
\be   
\hcXr_0=\cXsf_0 \prod_{\substack{\gamma_1\in\Gact	\\ \Re(Z_{\gamma_1})>0}}
\prod_{k=1-|\gamma_{01}|}^{|\gamma_{01}|-1}\prod_{n\in\IZ}
\Lambda\(\frac{Z_{\gamma}}{t_0},1-\[\vth_{\gamma_1}+(k+n)\alpha \]\)^{-s_{01}\Omega_n(\gamma_1) },
\label{hcXr-res}
\ee
This result coincides with \cite[Th.5.1]{Barbieri:2019yya} up to several discrepancies 
which can all be traced back to either differences in conventions or the way chosen to formulate the result:
\begin{itemize}
\item 
The different sign in the power is due to the opposite sign convention for the Dirac product and 
the additional sign factor appearing in the definition of the Laurent coefficients $\Omega_n(\gamma)$ in \cite{Barbieri:2019yya}.
 
\item 
The presence of the bracket $[\; \cdot\;]$ in the argument of the $\Lambda$-function, which
already appeared in the classical case \cite{Alexandrov:2021wxu} but is absent in \cite{Barbieri:2019yya}, 
simply means that the solution constructed in \cite{Barbieri:2019yya} should be seen as 
an analytic continuation of the branch where the bracket is constant. This can be ensured by restricting to 
$\Re\vth_\gamma\in[0,1)$ and $\Re \alpha=0$.

\item
The difference in the expression of the second argument of the $\Lambda$-function
appears because our variable $\vth_\gamma$ includes the quadratic refinement, whereas
$\theta(\gamma)$ in \cite{Barbieri:2019yya} does not.

\item
The last difference is that in our case the product over $k$ goes over a symmetric range, whereas in \cite{Barbieri:2019yya}
the range is shifted to positive or negative values depending on the sign $s_{01}$.
The difference would disappear if we wrote the solution \eqref{hcXr-res} using the star product, which precisely corresponds
to the way how the function found in \cite{Barbieri:2019yya} should be understood in our terms. 
\end{itemize}
Thus, our formalism successfully reproduces the solution of the simplest quantum RH problem.\footnote{Note that 
the function $\psi(t)$ \eqref{Jref}, which describes the adjoint form of the solution (see \eqref{mapPsi}),
evaluated in the uncoupled case appears to be {\it different} from 
the function representing the adjoint form that was found in \cite[Th.3.2]{Barbieri:2019yya}.
This is not surprising since the adjoint form is far from unique.
Instead, due to \eqref{qKSjump-psi}, it can be seen as a solution of the problem formulated in \cite[Remark 3.8]{Barbieri:2019yya}.}

\section{Conclusions}
\label{sec-conl}

In this paper we investigated the quantum RH problem induced by refined BPS indices.
Our main goal was to reformulate this problem in terms of a system of integral equations
involving the Moyal star product, which provide a non-commutative deformation of the TBA-like equations 
appearing in the classical RH problem.
We have succeeded in finding such a reformulation, though it provides a solution of the quantum RH problem 
only as a function of the solution of an integral equation. 
Unfortunately, we have not been able to find a sufficiently simple deformation of the classical TBA 
that would directly give the refined solution.

The main advantage of having the system of integral equations is that
it is amenable to a perturbative treatment and can be used to generate a formal series in BPS indices
that solves the equations and hence the quantum RH problem.
We have explicitly constructed such a perturbative expansion and showed that it has a smooth unrefined limit
where it gives rise to a new representation for the solution of the classical RH problem.
Whereas usually the solution of the classical TBA is expressed as a sum over all rooted trees, 
in the new representation only the so-called linear trees appear, which considerably simplifies the analysis.

Another interesting byproduct of this construction is that it automatically provides generating functions
of the solutions both at refined and unrefined levels.
This might have important implications. For instance, it opens a door for investigating various symmetries
since it is much easier to study a single function than quantities defined on a lattice.
Furthermore, we have derived perturbative expansions for the generating functions, which  
are also represented using linear trees and a sum over permutations.

As a check of our construction, we have shown that it reproduces the solution of the quantum RH problem found in 
\cite{Barbieri:2019yya} for the case of uncoupled BPS structures.
It would be interesting to extend this analysis to the case of the resolved conifold 
for which a solution has been found in \cite{Chuang:2022uey}.
However, based on the classical case analysis \cite{Alexandrov:2021prq},
one can expect that it would require finding proper prescriptions for dealing with divergent series.

In this paper we presented three setups that realize the quantum and classical RH problems.
Most of the results that we described are valid for two of them, whereas in Setup 2 relevant for \cite{Gaiotto:2008cd,Alexandrov:2008gh,Cecotti:2014wea},
the same construction produces a solution that does not have the unrefined limit.
We suggested a way to remedy this problem by noticing that in this setup the construction allows for an ambiguity. 
By requiring the existence of the unrefined  limit, 
we have fixed first terms in the expansion of the functional encoding the ambiguity and guessed its exact form. 
However, the unrefined limit of the resulting solution turns out to be different from the solution of the classical RH problem,
which can be traced back to the modified asymptotic conditions.
Thus, the fate of the quantum RH problem in this setup is still an open problem.

Finally, we hope that the analysis presented in this paper is a step towards understanding 
the non-commutative deformation induced by refinement of the moduli spaces arising in gauge and string theories.
Since these are typically quaternionic spaces, which are best described using twistor methods,
this might have intriguing connections to Penrose's twistor quantization program.
Furthermore, in the string theory context, the relevant spaces are QK
whose twistor spaces carry a holomorphic contact structure and, as far as we know, 
only recently an attempt was made to define its refined version \cite{Alexandrov:2023wdj}.
(Note that whereas in the gauge theory context the non-commutativity induced by the refinement 
has a physical interpretation in terms of a ring of line operators \cite{Gaiotto:2010be},
such an interpretation is lacking for string theory constructions.)
Other recent exciting developments, such as \cite{Brini:2024gtn,gräfnitz2025enumerativegeometryquantumperiods,Huang:2025xkc},
show that we are still at the beginning of exploring the vast and wonderful world of refined structures.

\section*{Acknowledgments}

We would like to thank Tom Bridgeland, Andrea Brini, Andy Neitzke and Boris Pioline for valuable discussions.
SA is grateful to the organizers of the ``Workshop on Resurgence, wall-crossing and geometry"
and ``Algebra and Quantum Geometry of BPS Quivers" at the SwissMAP Research Station
in Les Diablerets, where a part of this work was completed, 
for the kind hospitality.

\appendix

\section{Revisiting the CNV proposal}
\label{ap-qTBA}

In \cite{Cecotti:2014wea}, a quantum version of the TBA equation \eqref{TBAeq-sh} has been proposed
in the context of $N=2$ supersymmetric theories in the $\Omega$-background, corresponding to Setup 2 of section \ref{sec-TBA}. 
Using notations from the main text, 
it can be written as\footnote{In \cite{Cecotti:2014wea}, the first argument of the function $G_{\gamma_{12}}$
	is multiplied by $y^{2s_{\gamma_2}}$ where $s_\gamma$ are spins of BPS particles. 
	For simplicity, we restrict to scalars.}
\be
\hcX_0=T\left\{\exp\[-\frac{1}{4\pi\I}\,\IS_1 \,\frac{t_1 +t_0}{t_1-t_0}\, \log G_{\gamma_{01}} (\hcX_1;y) \]\hcXsf_0\right\}.
\label{qTBAeq}
\ee
Here the non-commutativity is induced by the angle variables
\be  
[\Theta_\gamma,\Theta_{\gamma'}]=\frac{\log y^2}{(2\pi\I)^2}\, \langle\gamma,\gamma'\rangle.
\ee 
This implies the commutation relation (cf. \eqref{starXX})
\be
\hcXsf_{\gamma}(t)\hcXsf_{\gamma'}(t') =
y^{2\langle\gamma,\gamma'\rangle}\hcXsf_{\gamma'}(t')\hcXsf_{\gamma}(t),
\label{comXsf}
\ee
and it is assumed that a similar relation holds for $\hcX_\gamma$'s evaluated at equal `times', where
time is associated with the phase of the twistor parameter $t$:
\be  
\hcX_\gamma(\rho e^{\I\tau})\, \hcX_{\gamma'}(\rho' e^{\I\tau})
=y^{2\langle\gamma,\gamma'\rangle}\hcX_{\gamma'}(\rho' e^{\I\tau})\,\hcX_\gamma(\rho e^{\I\tau}) .
\label{comm-cXg}
\ee 
The time ordering appearing in \eqref{qTBAeq} is then defined by
\be  
T\Bigl\{\hcX_\gamma(\rho e^{\I\tau})\, \hcX_{\gamma'}(\rho' e^{\I\tau'})\Bigr\}
=\left\{\begin{array}{ll}
	\hcX_\gamma(\rho e^{\I\tau})\, \hcX_{\gamma'}(\rho' e^{\I\tau'}),
	\qquad &\tau>\tau'
	\\
	y^{2\langle\gamma,\gamma'\rangle}\hcX_{\gamma'}(\rho' e^{\I\tau'})\,\hcX_\gamma(\rho e^{\I\tau}) ,
	\qquad &\tau'>\tau.
\end{array}
\right.
\label{def-Tprod}
\ee
Finally, the function $G_{\gamma_{01}} (\hcX_1;y)$ is given by the adjoint action of the quantum dilogarithm \eqref{defEy}:
\be  
E_y(\hcX_1) \, \hcX_0 \, E_y^{-1}(\hcX_1)=G_{\gamma_{01}} (\hcX_1;y) \hcX_0.
\ee 
Using the first representation in \eqref{defEy} of the quantum dilogarithm and the Campbell identity, 
it is straightforward to show that 
\be   
\log G_{\gamma_{01}} (\hcX_1;q)=\sum_{k=1}^\infty \frac{\(1-q^{k\gamma_{01}}\)\hcX_1^k}{k(q^{k/2}-q^{-k/2})}\, .
\label{funG}
\ee 

Although \cite{Cecotti:2014wea} writes the quantum TBA in terms of unrefined BPS indices $\Omega(\gamma)$,
hidden in \eqref{qTBAeq} in the symbol \tIS{i}, we believe that they should be replaced by refined indices $\Omega(\gamma,y)$.
Then the explicit expression \eqref{funG} and the rational refined BPS indices \eqref{defbOm} allow us
to rewrite the equation \eqref{qTBAeq} in a much simpler form, 
which is a straightforward non-commutative generalization of \eqref{TBAeq-sh},
\be
\hcX_0=T\left\{\exp\[\IS_1 \,\hK_{01} \hcX_1 \]\hcXsf_0\right\},
\label{qTBA}
\ee
where the kernel is given by
\be
\hK_{ij}=\frac{y^{2\gamma_{ij}}-1}{2(y-y^{-1})}\,\frac{t_j +t_i}{t_j-t_i}\, .
\ee

However, the question is whether the quantum TBA \eqref{qTBA} indeed provides a solution of the quantum RH problem.
There are two observations that seem to indicate that this is not the case.
First, the two commutation relations, \eqref{comXsf} and \eqref{comm-cXg}, 
appear to be mutually inconsistent provided $\hcX_\gamma$ solves \eqref{qTBA}. 
To see this, it is sufficient to substitute in the commutator of two $\hcX_\gamma$'s their expansion
obtained from the non-commutative TBA 
\be 
\hcX_0\approx \hcXsf_0 +\IS_1\hK_{01} T(\hcXsf_1\hcXsf_0)
+\IS_1\IS_2\(\hf\,\hK_{01}\hK_{02}+\hK_{01}\hK_{12}\) T(\hcXsf_2\hcXsf_1\hcXsf_0)+\cdots .
\ee 
Then one finds that, at equal parameters $t$, the commutation relation \eqref{comm-cXg}
is spoiled at quadratic order in BPS indices, whereas at non-equal parameters, even if they have equal phases,
it is spoiled already at linear order. 

The second observation is that the wall-crossing relation \eqref{qKSjump-X}, which seems to follow from \eqref{qTBA} 
in the same way as in the classical case, in fact, does not hold. 
This is because, after $t_0$ crosses an active BPS ray, the integral picks up the residue and one obtains 
an exponential of the sum of two operators. But due to the non-commutativity of the operators, 
it is {\it not} equal to the product of two exponentials, as required by the wall-crossing condition.\footnote{One 
	could think that the exponential of the sum is equal to the product of exponentials due to the time ordering
	appearing in \eqref{qTBA}. However, as it is defined in \eqref{def-Tprod}, it does not actually help
	because its effect is just to extend the commutation relation \eqref{comm-cXg} to non-equal `times'.
	It is possible that a proper modification of the $T$-product can resolve all problems, 
	but we did not investigate this possibility.}

\section{Wall-crossing of $\cXr_\gamma$}
\label{ap-wallcros}

To derive the wall-crossing behavior of $\cXr_\gamma$, 
we compute the effect of crossing an active BPS ray $\ell$ on the perturbative expansion \eqref{expXref}.
Namely, let $t_0$ cross $\ell$ in the clockwise direction.
Naively, it is sufficient to pick up the residue of the integral over $t_1$ at $t_1=t_0$ for $\gamma_1\in \Gamma_\ell$.
However, there are additional contributions due to the possibility of having several overlapping contours, 
which invokes the symmetric prescription. Let us consider a contribution to the $n$-th term on the r.h.s. of 
\eqref{expXref} where the first $m$ integrals go along $\ell$, i.e. $\gamma_k\in \Gamma_\ell$ for $1\leq k \leq m$.
According to the symmetric prescription, it is given by 
\be 
\hcXr_0^{(n,m)}=\(\frac{1}{m!}\sum_{\sigma\in S_m}
\prod_{j=1}^{m}\sum_{\gamma_j\in\Gamma_{\ell} }\,\int\limits_{\sigma(\ell_j)} \de t_j\)
\sum_{\gamma_{m+1}\in\Gamma_{\!\star}\backslash \Gamma_\ell }\,\int\limits_{\ell_{\gamma_{m+1}}}\de t_{m+1}
\(\prod_{j=m+2}^n\sum_{\gamma_j\in\Gamma_{\!\star} }\,\int\limits_{\ell_{\gamma_j}}\de t_j\)\cI^{(n)},
\label{contr-sigma}
\ee 
where $S_m$ is the group of permutations of $m$ elements, $\ell_j=e^{-\I j \eps}\ell$ 
with $\eps$ a positive infinitesimal parameter and $\cI^{(n)}$ denotes the integrand.
When one starts moving $t_0$ in the clockwise direction, one picks up a residue when it crosses $\sigma(\ell_1)$,
setting simultaneously $t_1=t_0$. Therefore, if $\sigma(2)>\sigma(1)$, moving $t_0$ further, the first residue
picks up a contribution of the pole from the integral over $t_2$, and so on.
Given that there are $\frac{m!}{k!(m-k)!}$ permutations with the first $k$ elements satisfying 
$\sigma(1)<\cdots <\sigma(k)$ and the remaining $m-k$ elements having a fixed ordering,
on the other side of the BPS ray the contribution \eqref{contr-sigma} is given by
\be 
\sum_{k=0}^m \frac{1}{k!}
\prod_{j=1}^{k}\[\sum_{\gamma_j\in\Gamma_{\ell} }\frac{\sigma_{\gamma_j}\bOm(\gamma_j,y)}{y-1/y}\,
y^{\sum_{i=0}^{j-1}\gamma_{ij}}\]
\hcXr_{\gamma_0+\gamma_{(k)}}^{(n-k,m-k)}(t_0^+),
\label{jump-sigma}
\ee 
where $\gamma_{(k)}=\sum_{j=1}^{k} \gamma_j$.
Then the sum over $n$ and $m$ replaces the last factor by 
$\cXr_{\gamma_0+\gamma_{(k)}}=y^{\langle\gamma_0,\gamma_{(k)}\rangle}\cXz_{\gamma_{(k)}}\star \cXr_{\gamma_0}$
where we used the relation \eqref{comm-Xr}.
Furthermore, since the charges belonging to $\Gamma_\ell$ are all mutually local, 
$\sum_{i=0}^{j-1}\gamma_{ij}=\gamma_{0j}$ and $\cXz_{\gamma_{(k)}}=\cXz_{\gamma_1}\star\cdots \star \cXz_{\gamma_{k}}$. 
As a result, the sum over $k$ gives rise to 
an exponential and one arrives at the wall-crossing relation \eqref{qKSjump-Xr}.

\section{Proofs}
\label{ap-proofs}

\lfig{A representation of the identity \eqref{cyc3} in terms of rooted trees.
	Here $T_k$'s are any rooted subtrees.}
{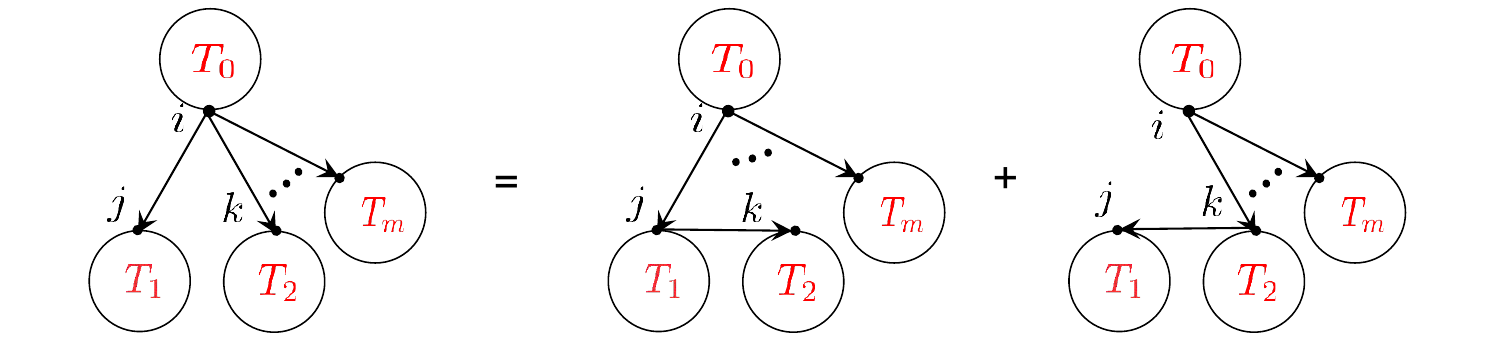}{15cm}{fig-identity}{-0.3cm}

In this appendix, we present proofs of several statements made in the main text,
which all have to do with various perturbative representations of the quantum and classical variables.
The main tool used to prove these statements is the identity \eqref{cyc3}.
Importantly, it has a geometric interpretation. The point is that, in all equations appearing in this work,
the kernels $\ker_{ij}$ can be thought of as factors assigned to edges of a labeled rooted tree (see, e.g., \eqref{Hexpand}).
Hence, we introduce 
\be 
\frKr_T=\prod_{e\in E_T}\ker_{s(e)t(e)}.
\label{def-frK}
\ee  
Then the identity \eqref{cyc3} is equivalent to the identity between the factors $\frKr_T$ associated 
to the labeled rooted trees shown in Fig. \ref{fig-identity}. 
It can be applied to any vertex that has at least two children and produces two trees 
where the number of children of that vertex decreased by one and the depth of some vertices increased by one.
It is clear that, applying this identity recursively, one can express $\frKr_T$ for any tree as 
a sum over linear trees for which every vertex has only one child. 
In fact, one can make an even stronger statement which will be extensively used below. 
To this end, let us recall that a rooted tree induces a natural partial ordering on its vertices: 
$v<v'$ if the path from $v'$ to the root passes through $v$.
Besides, we introduce the labeling map $\ell_T \ :\ V_T\to \IN$ from vertices of a tree $T$ to natural (unequal) numbers.
Then the following statement holds

\begin{lemma}\label{lemma-tree}
For a labeled rooted tree $T$ with $n$ vertices, one has 
\be 
\frKr_T=\sum_{T'\in\IT_{n}^{\rm lin}(T)} \frKr_{T'},
\label{eq-lemma}
\ee 
where $\IT_{n}^{\rm lin}(T)$ is the set of linear labeled rooted trees with $n$ vertices such that 
the labeling preserves the partial ordering of the original tree $T$ in the sense that
$v<v'\Rightarrow \ell_{T'}^{-1}(\ell_T(v))< \ell_{T'}^{-1}(\ell_T(v'))$.
\end{lemma}
\begin{proof}
At the first step, let us prove the statement of the lemma for a particular class of rooted trees given by 
``almost linear trees", namely linear trees consisting of $n-1$ vertices and the additional $n$-th vertex attached to 
the $k$-th vertex with $1\leq k\leq n-1$.\footnote{Of course, if $k=n-1$, this is the usual linear tree.}
Applying the basic identity, one obtains two trees: one is a linear tree from the set $\IT_{n}^{\rm lin}(T)$,
while the second is again an almost linear tree with the $n$-th vertex attached to the $k+1$ vertex
(see Fig. \ref{fig-almostlinear}).
It is obvious that proceeding in this way one arrives at the sum of $n-k$ terms represented by 
linear trees where the $n$-th vertex is inserted after the $m$-th vertex where $k\leq m\leq n-1$.
The set of these trees is exactly $\IT_{n}^{\rm lin}(T)$, which proves the statement.

\lfig{The almost linear tree appearing in the proof of Lemma \ref{lemma-tree} 
	and the result of application of the basic identity.}
{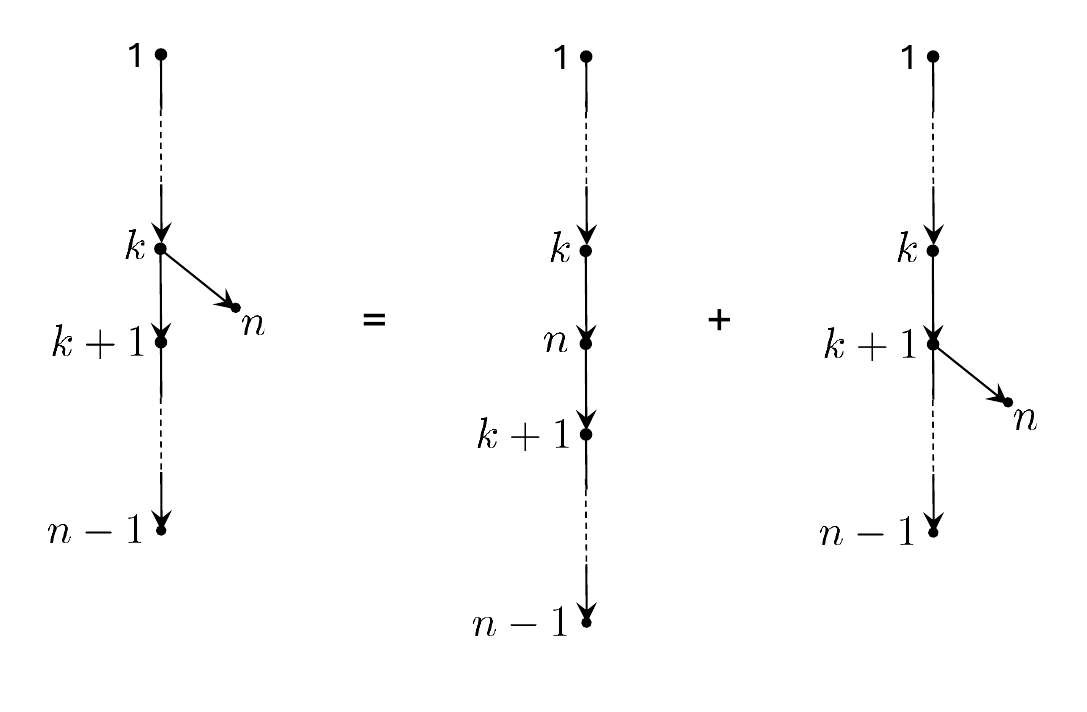}{10cm}{fig-almostlinear}{-0.8cm}

To generalize to arbitrary rooted tree, we proceed by induction.
For $n=1$ and 2 the statement is trivial since the r.h.s. of \eqref{eq-lemma} is identical to the l.h.s.
Let us assume now that \eqref{eq-lemma} holds for all trees with $n-1$ vertices and consider a tree with $n$ vertices.
Pick up any leaf and for simplicity consider a labeling where its label is $n$. Applying the induction hypothesis 
to the tree where this leaf is removed, one obtains the sum over almost linear trees where the $n$-th vertex is attached 
to the vertex of the linear subtree carrying the same label as its parent in the original tree.
Next, one can apply the relation \eqref{eq-lemma} to these almost linear trees as it was proven at the first step.
In this way one gets a sum over linear trees which all satisfy the conditions of the set $\IT_{n}^{\rm lin}(T)$.
Moreover, it is easy to see that they are all different and also exhaust this set, which proves the lemma.
\end{proof}

\subsection{Perturbative solution}
\label{ap-pert}

First, we prove the representation \eqref{linhHr-comm} for the perturbative solution of the quantum RH problem 
encoded in $\hcXr_\gamma(t)$. We proceed by induction in the perturbation order. 

In the linear approximation, $\psi(t_0)$ in \eqref{hcXranz} can be replaced by 
$1+\mbox{\resizebox{0.025\vsize}{!}{$\displaystyle{\IS_1}$}}\Kr_{01} \cXz_1$
and it is trivial to see that keeping only the first order in the expansion 
gives 
\be
\hcXr_0=\cXz_0+
\IS_1\ker_{01}\left\{\cXz_0,\cXz_1\right\}_\star+O(\bOm(\gamma,y)^2),
\label{linhHr-comm-1}
\ee
consistently with \eqref{linhHr-comm}.

So let us assume that the representation \eqref{linhHr-comm} holds up to order $n-1$.
To prove that it continues to hold at the next order, we note that the expansion of \eqref{hcXranz} 
in powers of $\cXr_\gamma$ leads to
\be
\hcXr_0 = \cXz_0 + \sum_{m=1}^{\infty} (-1)^{m-1}\IS_1\ker_{01}\[\prodi{k=2}{m}_\star \IS_k\Kr_{0k}\cXr_k \] \star
\left\{\cXz_0,\cXr_1\right\}_\star,
\label{expinref}
\ee
where $\prod_\star$ is defined in \eqref{defprodstar}.
In fact, the order of factors in the product over $k$ is not important, but we choose it be descending for definiteness.
Note that the product of kernels in the term of $m$-th order 
can be represented by a rooted tree with $m$ vertices labeled from 1 to $m$ 
and all connected to the root labeled by 0. Furthermore, if one assigns to the $k$-th vertex (except the root)
the factor $\mbox{\resizebox{0.025\vsize}{!}{$\displaystyle{\IS_k}$}}\,\cXr_k/(y-y^{-1})$
and accepts an additional rule that they are multiplied using the star product from left to right,
one arrives at the unique labeled rooted tree shown in Fig. \ref{fig-trees-pr1}(a).
Substituting further the expansion \eqref{expXref} of $\cXr_\gamma$ in terms of $\cXz_\gamma$, one arrives at
\be
\hcXr_0 = \cXz_0 + \sum_{n=1}^{\infty}\sum_{m=1}^{n} \frac{(-1)^{m-1}}{(y-y^{-1})^{n-1}}\[ \prod_{i=1}^n \IS_i\]
\sum_{\sum_{k=1}^m n_k=n}\frKr_{T_{\bfn}^{(m)}}
\[\prodi{k=2}{m}_\star \prodi{i=j_{k-1}+1}{j_k}^{\ \ \;\star}\cXz_i\] \star
\left\{\cXz_0,\prodi{i=1}{n_1}^\star\cXz_i\right\}_\star,
\label{expinsf}
\ee
where $n_k\geq 1$, $\bfn=(n_1,\dots n_m)$, $j_k=\sum_{l=1}^{k} n_l$, $T_{\bfn}^{(m)}$ is the tree shown in 
Fig. \ref{fig-trees-pr1}(b),
and the ordering of $\cXz_i$'s follows the labeling of the tree: from left to right and from top to bottom. 

\lfig{The rooted trees corresponding to the expansions of $\hcXr_0$ in $\cXr_0$ and $\cXz_0$, respectively.}
{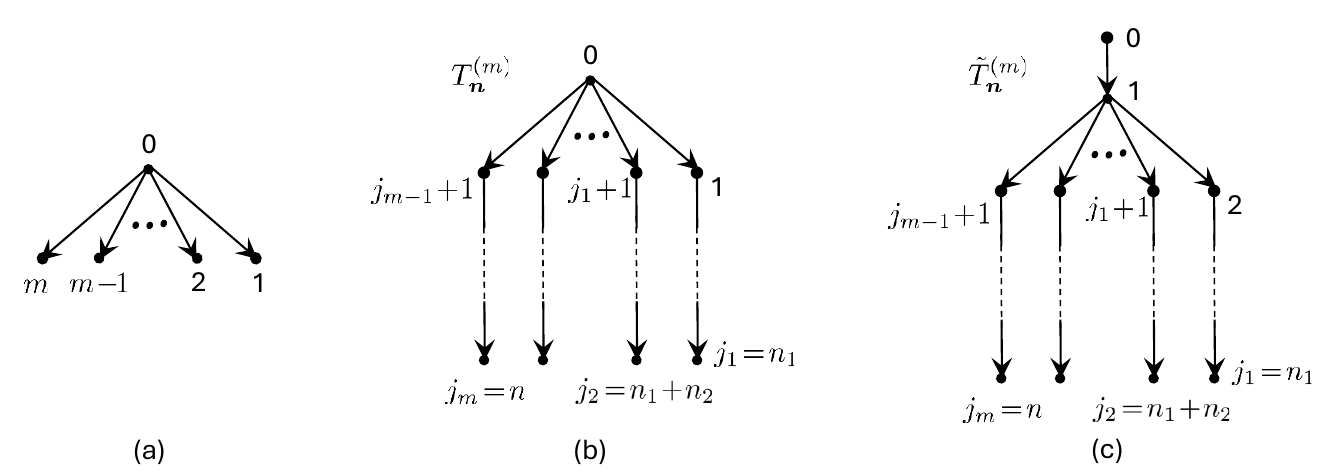}{16cm}{fig-trees-pr1}{-0.4cm}

The statement we want to prove would follow if we can show that the $n$-th term in \eqref{expinsf} 
is equal to 
\be  
\sum_{m=1}^{n-1} \frac{(-1)^{m-1}}{(y-y^{-1})^{n-2}} \[ \prod_{i=1}^n \IS_i\]
\sum_{\sum_{k=1}^m n_k=n}\frKr_{\tT_{\bfn}^{(m)}}
\[\prodi{k=2}{m}_\star\prodi{i=j_{k-1}+1}{j_k}^{\ \ \;\star}\cXz_{i}\] \star
\left\{\left\{\cXz_0,\cXz_1\right\}_\star,\prodi{i=2}{n_1}^\star\cXz_{i} \right\}_\star,
\label{expinsf-n}
\ee 
where $n_k\geq 1$ for $k\geq 2$ whereas $n_1\geq 2$, and $\tT_{\bfn}^{(m)}$ is the tree shown in Fig. \ref{fig-trees-pr1}(c).
Indeed, redefining $n_1\to n_1+1$, it is easy to see that 
this expression is the same as the term of order $n-1$ in \eqref{expinsf} with $\cXz_0$ replaced by 
$\mbox{\resizebox{0.025\vsize}{!}{$\displaystyle{\IS_1}$}}\ker_{01}\{\cXz_0,\cXz_1\}$
and therefore it is subject to the induction hypothesis, which immediately gives the $n$-th order term in \eqref{linhHr-comm}.

\lfig{Cancellation between trees after application of Lemma \ref{lemma-tree}.}
{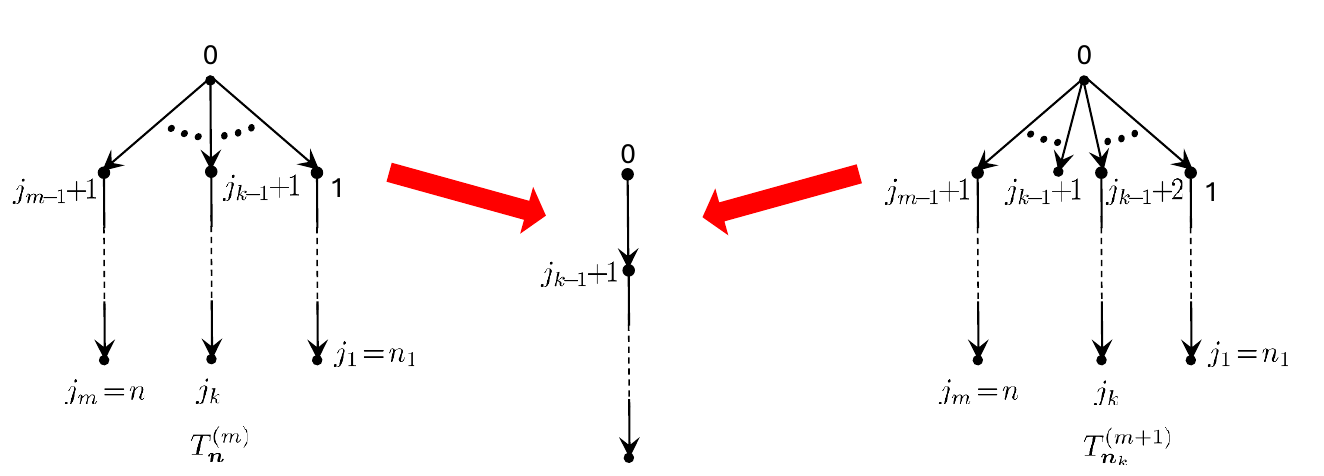}{16cm}{fig-trees-pr2}{-0.4cm}

To show the equality of \eqref{expinsf-n} and the $n$-th term in \eqref{expinsf}, we compare their contributions
after application to them of Lemma \ref{lemma-tree}. 
They are all given by linear trees that differ by the distribution of labels and $\cXz_i$-dependent factors determined 
by the structure of the original trees.
Let us consider linear trees produced by the tree $T_{\bfn}^{(m)}$ for some $\bfn$ with $m$ entries. 
According to Lemma \ref{lemma-tree},
the vertex attached to the root must carry the same label as one of the children of the root of $T_{\bfn}^{(m)}$.
Let it be $j_{k-1}+1$ with $k>1$ and assume that $n_k>1$. Then it is easy to realize that the corresponding contribution 
is canceled by the contribution of the same linear tree generated by $T_{\bfn_k}^{(m+1)}$ 
where $\bfn_k=(n_1,\dots,n_k-1,1,n_{k+1},\dots,n_m)$ with labels distributed as shown in Fig. \ref{fig-trees-pr2}.
They have the same $\cXz_i$-dependent factor, but come with opposite signs because $\bfn_k$ has $m+1$ entries. 
Note that the trees $T_{\bfn_k}^{(m+1)}$ constitute the subset of $T_{\bfn}^{(m+1)}$'s with $n_{k+1}=1$,
the missing case in the set considered above.
Thus, after all cancellations, we remain with (see Fig. \ref{fig-trees-pr3})
\begin{itemize}
	\item 
	the linear trees generated by all $T_{\bfn}^{(m)}$ with $n_1>1$ where the vertex attached to the root has label 1;
	\item 
	the linear trees generated by all $T_{\bfn_0}^{(m+1)}$ with $\bfn_0=(1,n_1-1,n_2,\dots,n_m)$
	where the vertex attached to the root has label 1;
	\item
	the linear trees generated by all $T_{\bfn_1}^{(m+1)}$ with $n_1>1$ where the vertex attached to the root has label $n_1$. 
\end{itemize}
Let us relabel the vertices of $T_{\bfn_1}^{(m+1)}$ generating the last remaining contributions as shown in Fig. \ref{fig-trees-pr3}
so that in the linear tree the vertex attached to the root has label 1 as in the other two cases.
Then it can be combined with the other remaining contributions leading to the following $\cXz_i$-dependent factor
\bea
&&
\[\prodi{k=2}{m}_\star \prodi{i=j_{k-1}+1}{j_k}^{\ \ \;\star}\cXz_i\] \star
\(\left\{\cXz_0,\prodi{i=1}{n_1}^\star\cXz_i\right\}_\star 
-\prodi{i=2}{n_1}^\star\cXz_i\star  \left\{\cXz_0,\cXz_1\right\}_\star
-\cXz_1 \star\left\{\cXz_0,\prodi{i=2}{n_1}^\star\cXz_i\right\}_\star\)
\nn\\
&=&\
\[\prodi{k=2}{m}_\star \prodi{i=j_{k-1}+1}{j_k}^{\ \ \;\star}\cXz_i\] \star
\left\{\left\{\cXz_0,\cXz_1\right\}_\star , \prodi{i=2}{n_1}^\star \cXz_i\right\}_\star.
\label{comb-proof1}
\eea
This is precisely the factor appearing in \eqref{expinsf-n}. Since in the resulting linear trees 
the label of the vertex attached to the root is always 1, through Lemma \ref{lemma-tree} they can be recombined into 
trees $\tT_{\bfn}^{(m)}$ shown in Fig. \ref{fig-trees-pr1}(c).
This proves the required statement.

\lfig{Contributions remaining after cancellations.}
{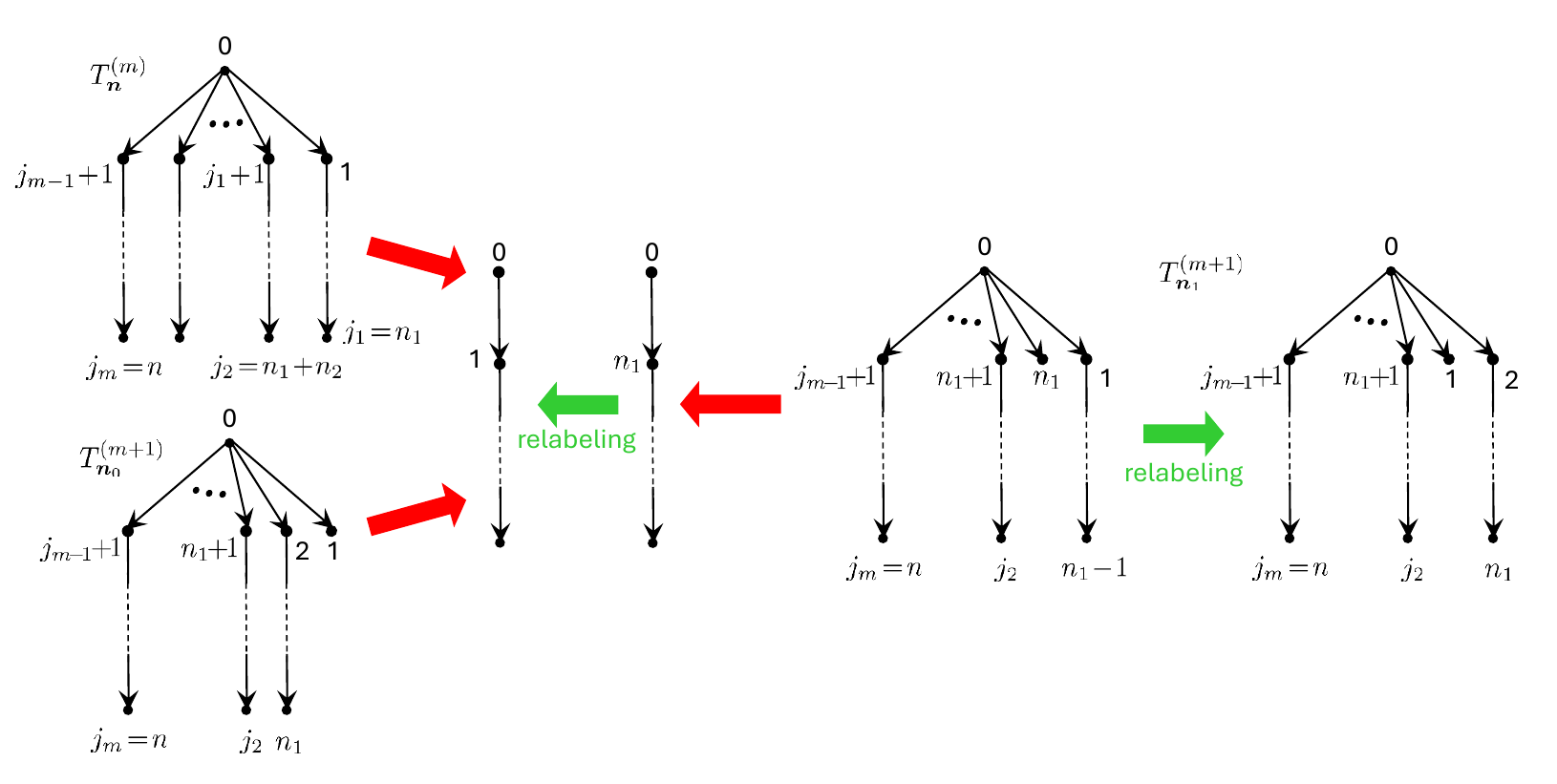}{18cm}{fig-trees-pr3}{-0.6cm}

\subsection{Unrefined solution}
\label{ap-unref}

Next, we prove the equivalence of the perturbative representations \eqref{Hexpand} and \eqref{compact-hcXr-lim}
of the unrefined functions $\cX_\gamma$.
As in the previous proof, we proceed by induction in the perturbation order. 

It is immediate to see that the terms with $n\leq 1$ in \eqref{compact-hcXr-lim} agree with the first two terms in
\eqref{Xexpandshort}. So let us assume that the two representations coincide up to order $n-1$ and 
compare the terms of order $n$.
To this end, let us note that the sum over rooted trees in \eqref{Hexpand} can be rewritten as a sum over 
{\it labeled} rooted trees, with the root carrying the fixed label 0. We denote the set of such trees with $n$ vertices by
$\IT_{n}^{\rm r,\ell}$. This affects only the weight of each tree, which now becomes $1/n!$ for all trees. 
Thus, the order $n$ contribution is given by
\be
\label{Hexpand-l}
\frac{1}{n!}\, \cXsf_0 \left(
\prod_{i=1}^n  \IS_i\,  \cXsf_i \right) \sum_{T\in \IT_{n+1}^{\rm r,\ell}}\,
\prod_{e\in E_T}K_{s(e)t(e)}\, .
\ee

Next, we apply Lemma \ref{lemma-tree} to the contribution of each tree in \eqref{Hexpand-l}.
It allows us to rewrite them as contributions of linear trees. 
Let us recombine the linear tree contributions that come from the rooted trees shown in Fig. \ref{fig-proof2}
and have the same labeling in which the vertex attached to the root has label 1.
These contributions have the same factor $\frKr_{T'}$, while their $\gamma_{ij}$-dependent factors give
\be 
\gamma_{01}\(\sum_{\cI\subseteq \Zv_m} \prod_{i\in \cI} \gamma_{0r_i}\prod_{j\in \Zv_m\setminus\cI} \gamma_{1r_j}\) 
\prod_{i=1}^m \gamma_{T_i}
=\gamma_{01}\prod_{i=1}^m (\gamma_{0r_i}+ \gamma_{1r_i}) 
\prod_{i=1}^m \gamma_{T_i},
\ee 
where $\Zv_{m}=\{1,\dots,m\}$, $r_i$ is the label of the root of the subtree $T_i$, and $\gamma_T=\prod_{e\in E_T}\gamma_{s(e)t(e)}$,
Thus, we get a linear tree contribution coming from a rooted tree with $n$ vertices 
where the root carries the charge $\gamma_0+\gamma_1$, the weight 
$\cXsf_0\mbox{\resizebox{0.025\vsize}{!}{$\displaystyle{\IS_1}$}}K_{01}\cXsf_1$ and its descendants 
are subtrees $T_i$. It is clear that all linear tree contributions with the child of the root labeled by 1
come from one of the trees shown in Fig. \ref{fig-proof2} for some subtrees $T_i$.
Moreover, if the child of the root has a different label, one can do a relabeling to make it 1.
Since the number of different labels is equal to $n$, this produces a factor of $n$
converting $1/n!$ in \eqref{Hexpand-l} to $1/(n-1)!$. 
As a result, applying Lemma \ref{lemma-tree} in the reverse direction, one obtains 
\be
\label{Hexpand-l2}
\frac{1}{(n-1)!} \, \cXsf_0\IS_1 K_{01}\cXsf_1 \left(
\prod_{i=2}^n  \IS_i\,  \cXsf_i \right) \sum_{T\in \IT_{n}^{\rm r,\ell}}\,
\prod_{e\in E_T}K_{s(e)t(e)}\, ,
\ee
where the root of $T$ carries the label 1 but charge $\gamma_0+\gamma_1$. This expression is subject to the induction hypothesis,
which leads to the representation \eqref{compact-hcXr-lim} upon replacing $\cXsf_i$ by $\cXz_i$.

\lfig{Trees constructed from the root 0, its child 1 and any number of subtrees $T_1, \dots, T_m$ which are attached 
	either to 0 or 1 in all possible ways.}
{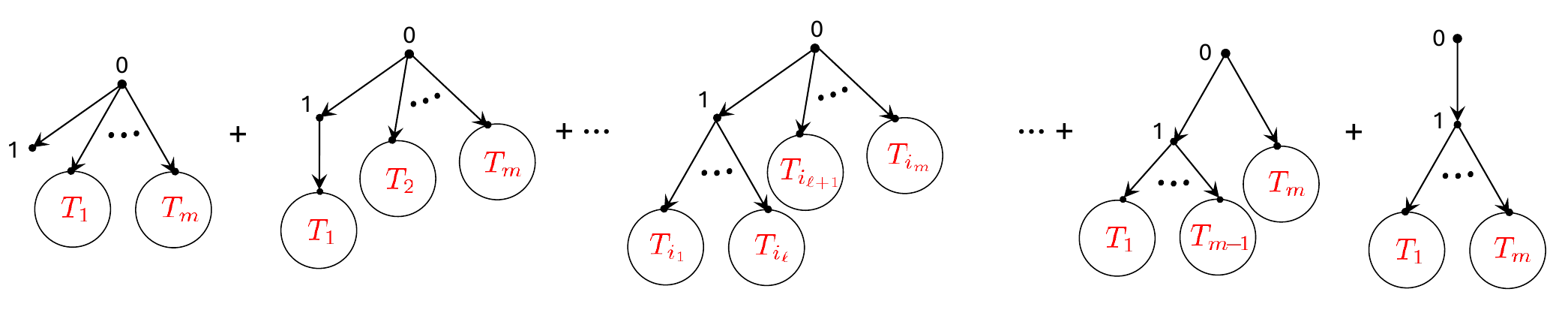}{18cm}{fig-proof2}{-0.4cm}

\subsection{Integral equation}
\label{ap-eq}

Here we prove the integral equation \eqref{intgeqn-comm}.
The proof is very similar to the proof of the representation \eqref{linhHr-comm} in appendix \ref{ap-pert}.
The idea is to start again from the expansion \eqref{expinref} of $\hcXr_\gamma$ in powers of $\cXr_\gamma$
but then, instead of the expansion \eqref{expXref} of $\cXr_\gamma$ in terms of $\cXz_\gamma$, 
to use its expansion \eqref{cXr-hcXr} in terms of $\hcXr_\gamma$.
Besides different variables, the only difference between the two expansions is the opposite order in which they are multiplied.
Therefore, the proof presented in \S\ref{ap-pert} goes through until Eq. \eqref{comb-proof1}
with the following changes:
\begin{itemize}
	\item 
	all $\cXz_i$ except $\cXz_0$ should be replaced by $\hcXr_i$;
	\item 
	the ordering of $\hcXr_i$'s follows the labeling of the tree according to the rule: from left to right 
	and from bottom to top (instead of from top to bottom before);
	\item 
	in the tree on the right in Fig. \ref{fig-trees-pr2}, the vertex $j_{k-1}+1$ and the subtree with labels $(j_{k-1}+2,\dots,j_k)$
	are swapped so that now $\bfn_k=(n_1,\dots,1,n_k-1,n_{k+1},\dots,n_m)$;
	\item 
	in the two trees on the right in Fig. \ref{fig-trees-pr3}, the child of the root without further descendants 
	(labeled by $n_1$ and 1, respectively) should appear as the leftmost.
\end{itemize}
This implies that the analogue of Eq. \eqref{comb-proof1} now takes the form
\bea
&&
\Pi_{2,m}^{\rm sf} \star
\(\left\{\cXz_0,\prodi{i=1}{n_1}_\star\hcXr_i\right\}_\star 
-\prodi{i=2}{n_1}_\star\hcXr_i\star  \left\{\cXz_0,\hcXr_1\right\}_\star\)
-\hcXr_1\star \Pi_{2,m}^{\rm sf}
\star\left\{\cXz_0,\prodi{i=2}{n_1}_\star\hcXr_i\right\}_\star
\nn\\
&=&\
\left\{\Pi_{2,m}^{\rm sf} \star 
\left\{\cXz_0,\prodi{i=2}{n_1}_\star \hcXr_i\right\}_\star ,\hcXr_1 \right\}_\star,
\label{comb-proof3}
\eea
where we denoted 
\be
\Pi_{2,m}^{\rm sf}=\prodi{k=2}{m}_\star \prodi{i=j_{k-1}+1}{j_k}_{\ \ \;\star}\hcXr_i.
\ee
The factor \eqref{comb-proof3} replaces a similar factor appearing in \eqref{expinsf-n}
and by induction leads to the desired equation \eqref{intgeqn-comm}.

\subsection{Generating function}
\label{ap-gen}

Finally, we prove the formula \eqref{logcJr2} for the perturbative expansion of the refined generating function $\cHr$.
Our starting point is the perturbative expansion obtained by expanding the logarithm in the definition \eqref{def-genref}
and then substituting the expansion \eqref{expXref} of $\cXr_\gamma$ in terms of $\cXz_\gamma$. 
This results in (cf. \eqref{expinsf})
\be
\cHr(t_0) = \sum_{n=1}^{\infty}\sum_{m=1}^{n} \frac{(-1)^{m-1}}{m(y-y^{-1})^{n-1}}\[ \prod_{i=1}^n \IS_i\]
\sum_{\sum_{k=1}^m n_k=n}\frKr_{T_{\bfn}^{(m)}}
\[\prodi{k=1}{m}^\star \prodi{i=j_{k-1}+1}{j_k}^{\ \ \;\star}\cXz_i\] ,
\label{expinsf-Hr}
\ee
where the set of trees and other notations are the same as in \S\ref{ap-pert}.
As in all previous cases, we apply Lemma \ref{lemma-tree} that expresses $\frKr_{T_{\bfn}^{(m)}}$ as a sum 
over labeled linear rooted trees and then relabel the vertices of the linear trees 
so that they are labeled in the ascending order. This ensures that all contributions have the same $t_i$-dependent factor
and differ only by the order in which the variables $\cXz_i$ are multiplied.
More precisely, one gets
\be
\cHr(t_0) = \sum_{n=1}^{\infty}\[ \prod_{k=1}^n \IS_k \cK_{k-1,k}\] \cP^{(n)},
\label{expinsf-Hr2} 
\ee
where
\be  
\cP^{(n)}=
\sum_{m=1}^{n}\frac{(-1)^{m-1}}{m(y-y^{-1})^{n-1}}
\sum_{\sum_{k=1}^m n_k=n}\sum_{\sigma\in S_{\bfn}}
\prodi{k=1}{m}^\star \prodi{i=j_{k-1}+1}{j_k}^{\ \ \;\star}\cXz_{\sigma(i)} 
\label{defPint}
\ee 
and $S_{\bfn}$ is the group of permutations of $n$ elements such that
$\sigma(j_{k-1}+1)<\cdots <\sigma(j_k)$ for all $k=1,\dots,m$.

The integrand $\cP^{(n)}$ can be rewritten as a sum over all permutations
\be  
\cP^{(n)}=
\sum_{\sigma\in S_n} C_\sigma\, \prodi{i=1}{n}^{\star}\cXz_{\sigma(i)} .
\label{Psigma}
\ee 
To get the coefficients $C_\sigma$, one can note that the order of the variables $\cXz_i$ in \eqref{defPint} remains 
unaffected if one applies the same $\sigma$ to a more refined partition than $\bfn=(n_1,\dots,n_m)$. 
In fact, the order depends only on the so called ascending runs, i.e. increasing contiguous subsequences of $\sigma$.
Let $s_\sigma$ is the number of such ascending runs and $l_a$, $a=1,\dots s_\sigma$, 
are  their lengths. From the restriction $\sigma\in S_{\bfn}$, it is clear that each ascending run
should consist of one or several subsets corresponding to linear subtrees in $T_{\bfn}^{(m)}$.
Let for the $a$-th run it is equal to $m_a$, which are restricted to satisfy $\sum_{a=1}^{s_\sigma} m_a=m$.
Then we conclude that
\be  
C_\sigma=-\sum_{m_a=1}^{l_a}\frac{(-1)^{\sum_{a=1}^{s_\sigma} m_a}}{\sum_{a=1}^{s_\sigma} m_a} 
\prod_{a=1}^{s_\sigma}\binom{l_a-1}{m_a-1},
\label{logJ1-csigma}
\ee 
where the first factor is nothing but the factor $(-1)^m/m$ in \eqref{defPint}, while the second factor 
is the number of trees $T_{\bfn}^{(m)}$ fitting the condition on $\sigma$.

To compute $C_\sigma$, let us consider the function
\be 
f_{s,\bfl}(x)=\prod_{a=1}^{s}\sum_{m_a=1}^{l_a} \binom{l_a-1}{m_a-1}x^{m_a}.
\label{def-funf}
\ee
It is related to $C_\sigma$ through an integral
\be 
C_\sigma=-\int_{0}^{-1} \frac{\de x}{x}\, f_{s_\sigma,\bfl}(x),
\label{relCf}
\ee
and can easily be evaluated as follows
\be 
f_{s,\bfl}(x)=x^s\prod_{a=1}^{s}\sum_{m_a=1}^{l_a-1} \binom{l_a-1}{m_a}x^{m_a}
=x^{s} \prod_{a=1}^{s} (1+x)^{l_a-1}
= x^{s}(1+x)^{n-s},
\ee
where we took into account that $\sum_{a=1}^{s} l_a=n$.
Using this in \eqref{relCf}, one finds that
\be
C_\sigma= (-1)^{s_\sigma+1}\, \frac{(s_\sigma-1)! (n-s_\sigma)!}{n!}\, .
\ee
Substituting this result into \eqref{Psigma} and $\cP^{(n)}$ into \eqref{expinsf-Hr2},
one obtains an explicit perturbative expansion of the refined generating function
\be
\cHr(t_0) = \sum_{n=1}^{\infty}\[ \prod_{k=1}^n \IS_k \cK_{k-1,k}\] 
\sum_{\sigma\in S_n} \frac{(-1)^{s_\sigma+1}}{\binom{n}{s_\sigma}s_\sigma }\, \prodi{i=1}{n}^{\star}\cXz_{\sigma(i)}.
\label{expinsf-Hr3} 
\ee

Remarkably, it can be further improved so that the existence of the unrefined limit of $\cHr$ becomes manifest. 
To this end, we use a relation proven in \cite[Lemma 6]{Solomon:1968}.
Applying it to \eqref{expinsf-Hr3}, one immediately obtains the representation \eqref{logcJr2}.

\

\noindent 
{\bf Data availability statement.}
No datasets were generated or analyzed during the current study.

\noindent 
{\bf Conflict of interest.}
The authors declare that there is no conflict of interest.

\providecommand{\href}[2]{#2}\begingroup\raggedright\endgroup


\begin{thebibliography}{10}
	
	\bibitem{Moore:1998et}
	G.~W. Moore, N.~Nekrasov, and S.~Shatashvili, ``{D particle bound states and
		generalized instantons},'' {\em Commun. Math. Phys.} {\bf 209} (2000) 77--95,
	\href{http://www.arXiv.org/abs/hep-th/9803265}{{\tt hep-th/9803265}}.
	
	\bibitem{Nekrasov:2002qd}
	N.~A. Nekrasov, ``{Seiberg-Witten prepotential from instanton counting},'' {\em
		Adv. Theor. Math. Phys.} {\bf 7} (2003), no.~5, 831--864,
	\href{http://www.arXiv.org/abs/hep-th/0206161}{{\tt hep-th/0206161}}.
	
	\bibitem{Gaiotto:2010be}
	D.~Gaiotto, G.~W. Moore, and A.~Neitzke, ``{Framed BPS States},'' {\em Adv.
		Theor. Math. Phys.} {\bf 17} (2013), no.~2, 241--397,
	\href{http://www.arXiv.org/abs/1006.0146}{{\tt 1006.0146}}.
	
	\bibitem{ks}
	M.~Kontsevich and Y.~Soibelman, ``{Stability structures, motivic
		Donaldson-Thomas invariants and cluster transformations},''
	\href{http://www.arXiv.org/abs/0811.2435}{{\tt 0811.2435}}.
	
	\bibitem{Dimofte:2009bv}
	T.~Dimofte and S.~Gukov, ``{Refined, Motivic, and Quantum},'' {\em Lett. Math.
		Phys.} {\bf 91} (2010) 1,
	\href{http://www.arXiv.org/abs/0904.1420}{{\tt 0904.1420}}.
	
	\bibitem{Dimofte:2009tm}
	T.~Dimofte, S.~Gukov, and Y.~Soibelman, ``{Quantum Wall Crossing in N=2 Gauge
		Theories},'' {\em Lett. Math. Phys.} {\bf 95} (2011) 1--25,
	\href{http://www.arXiv.org/abs/0912.1346}{{\tt 0912.1346}}.
	
	\bibitem{Barbieri:2019yya}
	A.~Barbieri, T.~Bridgeland, and J.~Stoppa, ``{A Quantized
		Riemann\textendash{}Hilbert Problem in Donaldson\textendash{}Thomas
		Theory},'' {\em Int. Math. Res. Not.} {\bf 2022} (2022), no.~5, 3417--3456,
	\href{http://www.arXiv.org/abs/1905.00748}{{\tt 1905.00748}}.
	
	\bibitem{Bridgeland:2016nqw}
	T.~Bridgeland, ``{Riemann-Hilbert problems from Donaldson-Thomas theory},''
	{\em Invent. Math.} {\bf 216} (2019) 69--124,
	\href{http://www.arXiv.org/abs/1611.03697}{{\tt 1611.03697}}.
	
	\bibitem{Bridgeland:2019fbi}
	T.~Bridgeland, ``{Geometry from Donaldson-Thomas invariants},''
	\href{http://www.arXiv.org/abs/1912.06504}{{\tt 1912.06504}}.
	
	\bibitem{Bridgeland:2020zjh}
	T.~Bridgeland and I.~A.~B. Strachan, ``{Complex hyperk{\"a}hler structures
		defined by Donaldson{\textendash}Thomas invariants},'' {\em Lett. Math.
		Phys.} {\bf 111} (2021), no.~2, 54,
	\href{http://www.arXiv.org/abs/2006.13059}{{\tt 2006.13059}}.
	
	\bibitem{Alexandrov:2021wxu}
	S.~Alexandrov and B.~Pioline, ``{Heavenly metrics, BPS indices and twistors},''
	{\em Lett. Math. Phys.} {\bf 111} (2021), no.~5, 116,
	\href{http://www.arXiv.org/abs/2104.10540}{{\tt 2104.10540}}.
	
	\bibitem{Gaiotto:2008cd}
	D.~Gaiotto, G.~W. Moore, and A.~Neitzke, ``{Four-dimensional wall-crossing via
		three-dimensional field theory},'' {\em Commun.Math.Phys.} {\bf 299} (2010)
	163--224, \href{http://www.arXiv.org/abs/0807.4723}{{\tt 0807.4723}}.
	
	\bibitem{Alexandrov:2008gh}
	S.~Alexandrov, B.~Pioline, F.~Saueressig, and S.~Vandoren, ``{D-instantons and
		twistors},'' {\em JHEP} {\bf 03} (2009) 044,
	\href{http://www.arXiv.org/abs/0812.4219}{{\tt 0812.4219}}.
	
	\bibitem{Alexandrov:2011va}
	S.~Alexandrov, ``{Twistor Approach to String Compactifications: a Review},''
	{\em Phys.Rept.} {\bf 522} (2013) 1--57,
	\href{http://www.arXiv.org/abs/1111.2892}{{\tt 1111.2892}}.
	
	\bibitem{Alexandrov:2009zh}
	S.~Alexandrov, ``{D-instantons and twistors: some exact results},'' {\em J.
		Phys.} {\bf A42} (2009) 335402,
	\href{http://www.arXiv.org/abs/0902.2761}{{\tt 0902.2761}}.
	
	\bibitem{Alexandrov:2010pp}
	S.~Alexandrov and P.~Roche, ``{TBA for non-perturbative moduli spaces},'' {\em
		JHEP} {\bf 1006} (2010) 066, \href{http://www.arXiv.org/abs/1003.3964}{{\tt
			1003.3964}}.
	
	\bibitem{Cecotti:2014zga}
	S.~Cecotti and M.~Del~Zotto, ``{$Y$ systems, $Q$ systems, and 4D
		$\mathcal{N}=2$ supersymmetric QFT},'' {\em J. Phys. A} {\bf 47} (2014),
	no.~47, 474001, \href{http://www.arXiv.org/abs/1403.7613}{{\tt 1403.7613}}.
	
	\bibitem{Ito:2017ypt}
	K.~Ito and H.~Shu, ``{ODE/IM correspondence and the Argyres-Douglas theory},''
	{\em JHEP} {\bf 08} (2017) 071,
	\href{http://www.arXiv.org/abs/1707.03596}{{\tt 1707.03596}}.
	
	\bibitem{Kachru:2018van}
	S.~Kachru, A.~Tripathy, and M.~Zimet, ``{K3 metrics from little string
		theory},'' \href{http://www.arXiv.org/abs/1810.10540}{{\tt 1810.10540}}.
	
	\bibitem{Alexandrov:2018lgp}
	S.~Alexandrov and B.~Pioline, ``{Black holes and higher depth mock modular
		forms},'' {\em Commun. Math. Phys.} {\bf 374} (2019), no.~2, 549--625,
	\href{http://www.arXiv.org/abs/1808.08479}{{\tt 1808.08479}}.
	
	\bibitem{DelMonte:2022kxh}
	F.~Del~Monte and P.~Longhi, ``{The threefold way to quantum periods: WKB, TBA
		equations and q-Painlev\'e},'' {\em SciPost Phys.} {\bf 15} (2023), no.~3,
	112, \href{http://www.arXiv.org/abs/2207.07135}{{\tt 2207.07135}}.
	
	\bibitem{Alday:2009dv}
	L.~F. Alday, D.~Gaiotto, and J.~Maldacena, ``{Thermodynamic Bubble Ansatz},''
	{\em JHEP} {\bf 1109} (2011) 032,
	\href{http://www.arXiv.org/abs/0911.4708}{{\tt 0911.4708}}.
	
	\bibitem{Alday:2010ku}
	L.~F. Alday, D.~Gaiotto, J.~Maldacena, A.~Sever, and P.~Vieira, ``{An Operator
		Product Expansion for Polygonal null Wilson Loops},'' {\em JHEP} {\bf 04}
	(2011) 088, \href{http://www.arXiv.org/abs/1006.2788}{{\tt 1006.2788}}.
	
	\bibitem{Cecotti:2014wea}
	S.~Cecotti, A.~Neitzke, and C.~Vafa, ``{Twistorial topological strings and a
		$\mathrm{tt}^*$ geometry for $\mathcal{N} = 2$ theories in $4d$},'' {\em Adv.
		Theor. Math. Phys.} {\bf 20} (2016) 193--312,
	\href{http://www.arXiv.org/abs/1412.4793}{{\tt 1412.4793}}.
	
	\bibitem{Alexandrov:2019rth}
	S.~Alexandrov, J.~Manschot, and B.~Pioline, ``{S-duality and refined BPS
		indices},'' {\em Commun. Math. Phys.} {\bf 380} (2020), no.~2, 755--810,
	\href{http://www.arXiv.org/abs/1910.03098}{{\tt 1910.03098}}.
	
	\bibitem{Alexandrov:2017qhn}
	S.~Alexandrov, S.~Banerjee, J.~Manschot, and B.~Pioline, ``{Multiple
		D3-instantons and mock modular forms II},'' {\em Commun. Math. Phys.} {\bf
		359} (2018), no.~1, 297--346,
	\href{http://www.arXiv.org/abs/1702.05497}{{\tt 1702.05497}}.
	
	\bibitem{Haydys}
	A.~Haydys, ``{Hyper-K\"ahler and quaternionic K\"ahler manifolds with
		$S^{1}$-symmetries},'' {\em J. Geom. Phys.} {\bf 58} (2008), no.~3, 293--306.
	
	\bibitem{Alexandrov:2011ac}
	S.~Alexandrov, D.~Persson, and B.~Pioline, ``{Wall-crossing, Rogers
		dilogarithm, and the QK/HK correspondence},'' {\em JHEP} {\bf 1112} (2011)
	027,
	\href{http://www.arXiv.org/abs/1110.0466}{{\tt 1110.0466}}.
	
	\bibitem{Hitchin:2012vvn}
	N.~Hitchin, ``{On the Hyperk{\"a}hler/Quaternion K{\"a}hler Correspondence},''
	{\em Commun. Math. Phys.} {\bf 324} (2013) 77--106,
	\href{http://www.arXiv.org/abs/1210.0424}{{\tt 1210.0424}}.
	
	\bibitem{Pioline:2009ia}
	B.~Pioline and S.~Vandoren, ``{Large D-instanton effects in string theory},''
	{\em JHEP} {\bf 07} (2009) 008,
	\href{http://www.arXiv.org/abs/0904.2303}{{\tt 0904.2303}}.
	
	\bibitem{Becker:1995kb}
	K.~Becker, M.~Becker, and A.~Strominger, ``Five-branes, membranes and
	nonperturbative string theory,'' {\em Nucl. Phys.} {\bf B456} (1995)
	130--152,
	\href{http://www.arXiv.org/abs/hep-th/9507158}{{\tt hep-th/9507158}}.
	
	\bibitem{Alexandrov:2010ca}
	S.~Alexandrov, D.~Persson, and B.~Pioline, ``{Fivebrane instantons, topological
		wave functions and hypermultiplet moduli spaces},'' {\em JHEP} {\bf 1103}
	(2011) 111, \href{http://www.arXiv.org/abs/1010.5792}{{\tt 1010.5792}}.
	
	\bibitem{Alexandrov:2014rca}
	S.~Alexandrov and S.~Banerjee, ``{Dualities and fivebrane instantons},'' {\em
		JHEP} {\bf 1411} (2014) 040,
	\href{http://www.arXiv.org/abs/1405.0291}{{\tt 1405.0291}}.
	
	\bibitem{Alexandrov:2023hiv}
	S.~Alexandrov and K.~Bendriss, ``{Hypermultiplet metric and NS5-instantons},''
	{\em JHEP} {\bf 01} (2024) 140,
	\href{http://www.arXiv.org/abs/2309.14440}{{\tt 2309.14440}}.
	
	\bibitem{Gaiotto:2014bza}
	D.~Gaiotto, ``{Opers and TBA},'' \href{http://www.arXiv.org/abs/1403.6137}{{\tt
			1403.6137}}.
	
	\bibitem{Alexandrov:2012au}
	S.~Alexandrov, J.~Manschot, and B.~Pioline, ``{D3-instantons, Mock Theta Series
		and Twistors},'' {\em JHEP} {\bf 1304} (2013) 002,
	\href{http://www.arXiv.org/abs/1207.1109}{{\tt 1207.1109}}.
	
	\bibitem{Bridgeland:2017vbr}
	T.~Bridgeland, ``{Riemann--Hilbert problems for the resolved conifold and
		non-perturbative partition functions},'' {\em Journal of Differential
		Geometry} {\bf 115} (3, 2020) 395--435,
	\href{http://www.arXiv.org/abs/1703.02776}{{\tt 1703.02776}}.
	
	\bibitem{Alexandrov:2021prq}
	S.~Alexandrov and B.~Pioline, ``{Conformal TBA for Resolved Conifolds},'' {\em
		Annales Henri Poincare} {\bf 23} (2022), no.~6, 1909--1949,
	\href{http://www.arXiv.org/abs/2106.12006}{{\tt 2106.12006}}.
	
	\bibitem{Filippini:2014sza}
	S.~A. Filippini, M.~Garcia-Fernandez, and J.~Stoppa, ``{Stability data,
		irregular connections and tropical curves},'' {\em Selecta Mathematica} {\bf
		23} (3, 2017) 1355--1418, \href{http://www.arXiv.org/abs/1403.7404}{{\tt
			1403.7404}}.
	
	\bibitem{Alexandrov:2025sig}
	S.~Alexandrov, ``{Mock modularity at work, or black holes in a forest},'' {\em
		Entropy} {\bf 27} (2025) 719, \href{http://www.arXiv.org/abs/2505.02572}{{\tt
			2505.02572}}.
	
	\bibitem{Chuang:2022uey}
	W.-Y. Chuang, ``{Quantum Riemann-Hilbert problems for the resolved conifold},''
	{\em J. Geom. Phys.} {\bf 190} (2023) 104860,
	\href{http://www.arXiv.org/abs/2203.00294}{{\tt 2203.00294}}.
	
	\bibitem{Manschot:2010qz}
	J.~Manschot, B.~Pioline, and A.~Sen, ``{Wall Crossing from Boltzmann Black Hole
		Halos},'' {\em JHEP} {\bf 1107} (2011) 059,
	\href{http://www.arXiv.org/abs/1011.1258}{{\tt 1011.1258}}.
	
	\bibitem{Alexandrov:2018iao}
	S.~Alexandrov and B.~Pioline, ``{Attractor flow trees, BPS indices and
		quivers},'' {\em Adv. Theor. Math. Phys.} {\bf 23} (2019), no.~3, 627--699,
	\href{http://www.arXiv.org/abs/1804.06928}{{\tt 1804.06928}}.
	
	\bibitem{Strachan:1992em}
	I.~A.~B. Strachan, ``{The Moyal algebra and integrable deformations of the
		selfdual Einstein equations},'' {\em Phys. Lett. B} {\bf 283} (1992) 63--66.
	
	\bibitem{Takasaki:1992jf}
	K.~Takasaki, ``{Dressing operator approach to Moyal algebraic deformation of
		selfdual gravity},'' {\em J. Geom. Phys.} {\bf 14} (1994) 111--120,
	\href{http://www.arXiv.org/abs/hep-th/9212103}{{\tt hep-th/9212103}}.
	
	\bibitem{Alexandrov:2023wdj}
	S.~Alexandrov, M.~Mari\~no, and B.~Pioline, ``{Resurgence of Refined
		Topological Strings and Dual Partition Functions},'' {\em SIGMA} {\bf 20}
	(2024) 073, \href{http://www.arXiv.org/abs/2311.17638}{{\tt 2311.17638}}.
	
	\bibitem{Brini:2024gtn}
	A.~Brini and Y.~Schuler, ``{Refined Gromov-Witten invariants},''
	\href{http://www.arXiv.org/abs/2410.00118}{{\tt 2410.00118}}.
	
	\bibitem{gräfnitz2025enumerativegeometryquantumperiods}
	T.~Gräfnitz, H.~Ruddat, E.~Zaslow, and B.~Zhou, ``Enumerative geometry of
	quantum periods,'' \href{http://www.arXiv.org/abs/2502.19408}{{\tt
			2502.19408}}.
	
	\bibitem{Huang:2025xkc}
	M.~Huang, S.~Katz, A.~Klemm, and X.~Wang, ``{Refined BPS numbers on compact
		Calabi-Yau 3-folds from Wilson loops},''
	\href{http://www.arXiv.org/abs/2503.16270}{{\tt 2503.16270}}.
	
	\bibitem{Solomon:1968}
	L.~Solomon, ``On the poincaré-birkhoff-witt theorem,'' {\em Journal of
		Combinatorial Theory} {\bf 4} (1968), no.~4, 363--375.
	
\end{thebibliography}

\end{document}